 \documentclass[final,3p,times,authoryear]{elsarticle}

\usepackage{amssymb}
\usepackage{amsmath}

\usepackage{bbm}
\usepackage{nicefrac}
\usepackage{tabularx} 
\usepackage{amsthm}
\usepackage{url}
\usepackage{color}
\usepackage{enumitem}
\newcommand{\mycitep}[2]{\citeauthor{#1}~(\citeyear{#1},~p.~#2)}

\newcommand{\myCite}[1]{\citeauthor{#1}~(\citeyear{#1})}
\newcommand{\mycite}[1]{\citeauthor{#1}~(\citeyear{#1})}
\newcommand{\citeI}[1]{\citeauthor{#1},~\citeyear{#1}}

\newcommand{\hyp}{$\text{H}_0$}

\newcommand{\dd}{\text{d}}
\newcommand{\bb}[1]{\boldsymbol{#1}}

\newcommand{\Var}{\text{Var}}
\newcommand{\ew}{\text{E}}

\newcommand{\R}{\mathbb{R}}
\newcommand{\mutilde}{\tilde{\mu}_{h,k}}
\newcommand{\nutilde}{\tilde{\nu}_{h,k}}
\newcommand{\muhat}{\hat{\mu}_{h,k}}
\newcommand{\nuhat}{\hat{\nu}_{h,k}}
\newcommand{\doublesumBlock}{\underset{i, j \, \in \, \mathcal{I}_{h,k}}{\; \sum \sum \;}}
\renewcommand{\P}{\text{P}}

\newtheorem{theorem}{Theorem}
\newtheorem{remark}[theorem]{Remark}
\newtheorem{proposition}[theorem]{Proposition}
\newtheorem{lemma}[theorem]{Lemma}
\newtheorem{assumption}{Assumption}

\journal{Journal of Multivariate Analysis}

\begin{document}

\begin{frontmatter}



\title{Detecting changes in the mean of spatial random fields on a regular grid} 


\author{Sheila Görz} 
\ead{sheila.goerz@tu-dortmund.de}
\author{Roland Fried}

\affiliation{organization={Department of Statistics, TU Dortmund University},
            addressline={Vogelpothsweg 87}, 
            city={Dortmund},
            postcode={44227}, 
            state={North Rhine-Westphalia},
            country={Germany}}

\begin{abstract}
We propose statistical procedures for detecting changes in the mean of spatial random fields observed on regular grids. The proposed framework provides a general approach to change detection in spatial processes. Extending a block-based method originally developed for time series, we introduce two test statistics, one based on Gini’s mean difference and a novel variance-based variant. Under mild moment conditions, we derive asymptotic normality of the variance-based statistic and prove its consistency against almost all non-constant mean functions (in a sense of positive Lebesgue measure). To accommodate spatial dependence, we {\color[RGB]{0, 0, 0} modify our procedures for $M$-dependent data and we} further develop a de-correlation algorithm based on estimated autocovariances. Monte Carlo simulations demonstrate that the tests maintain appropriate size and power for both independent and dependent data. In an application to satellite images, especially our variance-based test reliably detects regions undergoing deforestation.
\end{abstract}



\end{frontmatter}


\section{Introduction}

Abrupt changes in the structure of data can occur not only in time series but also in spatial random fields. Applications include satellite images to detect changes in nature, 
medical data 
and quality control. 
In all cases, it is important that structural breaks are recognized reliably. We consider random fields $(X_{\bb{i}})$, $\bb{1} \leq \bb{i} \leq \bb{n}, \, \bb{n} \in \mathbb{N}^d,$ stemming from {\color[RGB]{0, 0, 0} $d$-dimensional arrays} that follow the widely used signal-plus-noise model
\begin{align} \label{mod-sigNoise}
    X_{\bb{i}} = \mu\left(\frac{\bb{i}}{\bb{n}}\right) + Y_{\bb{i}},
\end{align}
{\color[RGB]{0, 0, 0} where $\frac{\bb{i}}{\bb{n}}$ is meant component-wise.} Here, $(Y_{\bb{i}})_{\bb{i} \in \mathbb{N}^d}$ is a stationary random field with $\ew(Y_{\bb{i}}) = 0$, $\Var(Y_{\bb{i}}) = \sigma^2$ and a continuous distribution. Our aim is to investigate whether the location function $\mu$ is constant or not. For time series, this is a well-studied problem; see e.g. \myCite{csorgo1997limit} for a survey. Many of these methods cannot be directly transferred to data over a higher-dimensional grid, as observations do not possess a natural order. Existing methods for spatial data often focus on one special type of change that should be detected. \myCite{fuentes2005formal} develops a method using spatial spectral analysis where she tests if there is interaction between space and frequency through a classical ANOVA. Changes in the mean cannot be detected with this method, and it lacks asymptotic results for a growing sample size. \mycite{gromenko2017detection} {\color[RGB]{0, 0, 0}investigate} changes in the mean of observations taken at multiple spatially correlated locations. In their work, the term change refers to a change over time rather than in space. \myCite{bucchia2014testing} tests for a change over a multi-dimensional "rectangle" by comparing the increase over a rectangle to the increase over the whole random field. This method needs maximization over all possible blocks and is computationally very intensive. \myCite{otto2016detection} and \myCite{kirch2025scan} develop methods for the detection of very specific change regions, the shape of which must be known in advance. The methods of Otto and Schmid can detect at most one change region, whereas the method of \myCite{kirch2025scan} can handle multiple change regions, but (asymptotic) critical values are not necessarily analytically known. \mycite{steland2025detection} addresses these issues by developing a test that employs Gumbel-type extreme value theory. However, the convergence of  such maximum-type statistics to their asymptotical distribution can be rather slow, so that large sample sizes are needed. Another method, proposed by \myCite{zhang2019spatial}, focuses more on precise localization of the change region(s), but it is only applicable for independent data and abrupt changes opposed to trends. By construction, the method does not give a sensible output if no change region is present. In an application to quality control of products, \mycite{jiang2005liquid} and \mycite{amirkhani2020novel} use ANOVA-based control charts to monitor product images. \mycite{okhrin2020new} and \mycite{okhrin2025monitoring} developed methods for monitoring changes in regular grids or images while using regions of interest for dimension reduction. These methods rely on reference images and require parameters to be estimated by a pre-run. Similarly, the technique of \mycite{mayrhofer2025robust} requires a pre-run for outlier detection in sequences of images. All the above monitoring procedures only allow for normally distributed errors and do not yield any asymptotics for growing images sizes, e.g. due to finer sampling. 

Our approach is based on the work of \myCite{schmidt2024detecting}, where blocks of a time series are compared to each other to find deviations in their location. In this study, we initially only consider the case where the $Y_i$ from model~(\ref{mod-sigNoise}) are iid. A possible non-stationarity in the mean, i.e. a change in location, is described by the function $\mu: [0, 1]^d \rightarrow \mathbb{R}.$ Our work extends Schmidt's test to two-dimensional data ($d = 2$), and a transfer to higher-dimensional data is straightforward. In extending the test method, we focus on the elementary scenario of independent, homoscedastic data. However, unlike the original test statistic, we consider more options for comparing the block means. 

The remainder of the paper is structured as follows. Section~\ref{sec-teststat} introduces our basic assumptions and the basic test statistic. Building on this, Section~\ref{sec-proof} introduces a variance-based variation to the test statistic from Section~\ref{sec-teststat}, and proves its convergence and consistency against almost all non-constant mean functions. Section~\ref{sec-decorr} discusses extensions to dependent data. To this end, we introduce {\color[RGB]{0, 0, 0} a modified test for $M$-dependent data and} a simple de-correlation algorithm {\color[RGB]{0, 0, 0} for the general case}. The results of a simulation study for both independent and dependent data are presented in Section~\ref{sec-sim}, and an application to satellite data of the Amazon rainforest is given in Section~\ref{sec-satellite}. Section~\ref{sec-summary} provides a summary and an outlook.

\section{Data situation and original test statistic} \label{sec-teststat}

We focus on the situation of 2-dimensional random fields {\color[RGB]{0, 0, 0} with independent variables,} observed on a regular grid, i.e., in our case $\bb{n} = (n, m)$. The data model then reads 
$$X_{i,j} = X_{i,j}^{(n,m)} = \mu\left(\frac{i}{n}, \frac{j}{m}\right) + Y_{i,j}, \quad i \in \{1, ..., n\}, \; j \in \{1, ..., m\},$$
where the observations $\left(X^{(n,m)}_{i,j}\right)$ stem from a double array {\color[RGB]{0, 0, 0}and the noise variables $(Y_{i,j})$ are i.i.d. with $\ew(Y_{i,j}) = 0$, $\Var(Y_{i,j}) = \sigma^2$}. We assume the following about the location function $\mu$:
\begin{assumption}\label{as-locFun}
    The location function $\mu: [0, 1]^2 \rightarrow \R$ is of the form: $\mu = \mu' + \sum_{i = 1}^K c_i \cdot \mathbbm{1}_{\mathcal{C}_i}$ where $\mu'$ is a continuous function. These sets $\mathcal{C}_i$ or their union $\mathcal{C} = \bigcup_{i = 1}^K \mathcal{C}_i$ will be called "change region(s)", their complement $\mathcal{B} := [0, 1]^2 \setminus \mathcal{C}$ will be called "base region". We assume
    \begin{enumerate}[label=(\alph*)]
        \item $\mathcal{C}_i \subset [0, 1]^2$, $i = 1, ..., K$, are finitely many disjoint Borel sets that contain the indices over which a location shift of magnitude $c_i$ occurs.
        \item All $\mathcal{C}_i$, as well as $\mathcal{B}$, have positive Lebesgue measure.
        \item The boundary of each $\mathcal{C}_i$ has a Lebesgue measure 0 and a finite length $\mathcal{L}_i$. As the number of change regions $K$ is finite, the total length of all boundaries $\mathcal{L}$ is also finite. 
     \end{enumerate}
\end{assumption}
Using this notation, we want to test the hypothesis pair
$$\mathbb{H}_0: \mu \text{ is constant} \quad \text{vs.} \quad \mathbb{H}_1: \mu \text{ is not constant}.$$
Note that it is possible to have no change regions $\mathcal{C}_i$ at all, even under the alternative. For example, if there is a trend present in $\mu$, the function can be continuous but not constant. {\color[RGB]{0, 0, 0}Under the hypothesis, it is trivial that there are no change regions.}
Adapting the test statistic of \myCite{schmidt2024detecting} to our spatial data setting leads us to two main components.\\
First, the given data is divided into $b_n \times b_m$ blocks of length $l_n \times l_m$. {\color[RGB]{0, 0, 0}We assume the following about the number of blocks and their length:
\begin{assumption} \label{as-blockSize}
    Each block has the dimension $l_n \times l_m$, $l_n = n^{s_1}, \; l_m = m^{s_2}$, $s_1, s_2 \in (0, 1)$. This results in $n^{1 - s_1} \times m^{1 - s_2} = b_n \times b_m$ blocks in total.
\end{assumption}}
Then, a statistic that represents the respective block adequately is taken. The most obvious choice for this is the arithmetic mean:
\begin{align*} 
    \hat{\mu}_{h,k} := \hat{\mu}_{h,k}^{(n,m)} := \frac{1}{l_n l_m} \sum_{i = (h - 1) l_n + 1}^{h l_n} \sum_{j = (k - 1) l_m + 1}^{k l_m} X_{i,j}, \quad h = 1, ..., b_n, \, k = 1, ..., b_m.
\end{align*}
A big advantage of the arithmetic mean is that its limit distribution is known and usually easy to work with. The block means form a double array
$$
 \left(\hat{\mu}_{h,k}^{(n,m)}\right) := \left(\hat{\mu}_{h,k}^{(n,m)}, \; h = 1, ..., b_n, \; k = 1, ..., b_m, \; n,m \in \mathbb{N}\right)
$$
since for increasing $n$ and $m$, the blocks may contain different $X_{i,j}$s. Nevertheless, the elements of the sequence $\left(\hat{\mu}_{h,k}^{(n,m)}\right)_{h,k}$ are independent for fixed $n, m$ {\color[RGB]{0, 0, 0} according to our  basic assumptions.}\\

In the second step, we apply a measure that is able to uncover possible differences between the block representatives. \myCite{schmidt2024detecting} uses \textbf{Gini's mean difference (GMD)}, which in our 2-dim. scenario reads as: 
\begin{align*} 
    U({\color[RGB]{0, 0, 0}n,m}) = \frac{1}{b_n b_m (b_n b_m - 1)} \sum_{h = 1}^{b_n} \sum_{k = 1}^{b_m} \sum_{h' = 1}^{b_n} \sum_{k' = 1}^{b_m} |\hat{\mu}_{h,k} - \hat{\mu}_{h',k'}|.
\end{align*}
Appropriately scaled, the following holds for the test statistic $U$ if $\hat{\mu}_{h,k}$ are the arithmetic block means:
\begin{theorem}
    Let $\ew(|Y_{1,1}|^{2 + \varepsilon}) < \infty$ for some $\varepsilon > 0 $ and $\color[RGB]{0, 0, 0}s_1, s_2 \in (0.5, 1)$. Then it holds under the assumption of a constant mean that
$${\color[RGB]{0, 0, 0}T^{GMD}_{n,m}}(\bb{X}) = \sqrt{b_n b_m}\left(\frac{\sqrt{l_n l_m}}{\hat{\sigma}} U({\color[RGB]{0, 0, 0}n,m}) - \frac{2}{\sqrt{\pi}}\right) \overset{\mathcal{D}}{\longrightarrow} \mathcal{N}\left(0, \frac{4}{3} + \frac{8}{\pi}\left(\sqrt{3} - 2\right)\right)$$
with $\hat{\sigma}$ being a {\color[RGB]{0, 0, 0}consistent} estimator for the standard {\color[RGB]{0, 0, 0}deviation }of $X_{i,j}^{(n,m)}$. 
\end{theorem}
Both {\color[RGB]{0, 0, 0}this asymptotics under the hypothesis} and the consistency of the test against almost all non-constant mean functions can be deduced from the proofs of \mycite{schmidt2024detecting} for the one-dimensional case. A similar test for constancy of the variance in time series was proposed by \mycite{schmidt2021asymptotic}. \\

As the arithmetic block means should asymptotically be normally distributed, regardless of the noise distribution, using the empirical variance instead of Gini's mean difference in the second step could be more efficient. Assuming the underlying data to be normally distributed, i.e. in our case $Y_{i,j} \sim \mathcal{N}(0, \sigma^2)$, and that there is a fixed number of blocks $2 \leq b_{(1)}, b_{(2)} < n,m$ in each dimension, this would result in an analysis of variance (ANOVA) test to check for variability between the blocks. But since, in the classical ANOVA, the number of blocks is fixed, there are location shifts that would not be detected if the number of blocks did not increase with growing sample size respectively finer sampling. The exemplary alternative in Figure~\ref{fig-alt5} demonstrates this problem graphically.\\
\begin{figure}[ht]
    \centering
    \includegraphics[width = 0.3\textwidth]{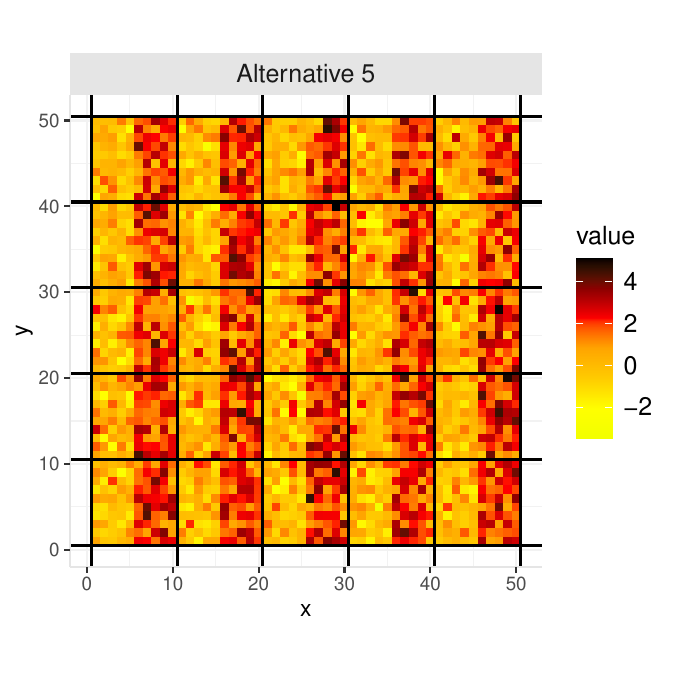}
    \caption{Example of an alternative that would not be detected if the number of blocks did not increase with growing sample size}
    \label{fig-alt5}
\end{figure}

Since, to the best of our knowledge, there are no asymptotics or modifications to cater non-normally distributed data or triangular arrays available for the ANOVA, in the next section we present a new, variance-based test statistic. We then prove its asymptotic convergence and its consistency against most alternatives. 

\section{A variance-based test} \label{sec-proof}

Instead of Gini's mean difference, we measure the variability between block means in a test based on the ideas of \myCite{schmidt2024detecting} by their variance:
$$\frac{1}{b_n b_m} \sum_{h = 1}^{b_n} \sum_{k = 1}^{b_m} (\hat{\mu}_{h,k} - \bar{X}_{n,m})^2 = \frac{1}{b_n b_m} \sum_{h = 1}^{b_n} \sum_{k = 1}^{b_m} \left(\hat{\mu}_{h,k} - \bar{\hat{\mu}}\right)^2,$$
{\color[RGB]{0, 0, 0}where $\bar{X}_{n,m} = \bar{X} = \frac{1}{nm} \sum_{i = 1}^n \sum_{j = 1}^m X_{i,j}$ is the arithmetic mean of all observations in our sample of size $nm$}. This corresponds to the numerator of an ANOVA. In order to apply a central limit theorem (CLT), we have to scale the statistic appropriately. In the following subsections, we first introduce the final test statistic, then follows the proof of its convergence under \hyp, starting with the notation used throughout the proof. Finally, we prove the consistency of the test.

\begin{theorem}\label{th-LyapunovCLTfull}
    Let {\color[RGB]{0, 0, 0}Assumption~\ref{as-blockSize}} be fulfilled and assume that $\ew(|Y_{i,j}|^{4 + \varepsilon}) < \infty$ for some $\varepsilon > 0$. Let $\hat{\sigma}$ be a (weakly) consistent estimator for $\sigma$. Then it holds under the hypothesis of a constant mean that
    \begin{align}\label{eq-LyapunovCLTfull}
    \begin{split}
    {\color[RGB]{0, 0, 0}T^{Var}_{n,m}}(X) = \frac{1}{\sqrt{2 b_n b_m}} \left[\frac{l_n l_m}{\hat{\sigma}^2} \left(\sum_{h = 1}^{b_n} \sum_{k = 1}^{b_m} \left(\hat{\mu}^{(n,m)}_{h,k}\right)^2 - b_n b_m \bar{X}_{n,m}^2 \right) - b_n b_m + 1\right] \overset{\mathcal{D}}{\longrightarrow}\mathcal{N}(0, 1) \quad   \text{ as } \, nm \rightarrow \infty.
    \end{split}
    \end{align}
\end{theorem}

The resulting test for structural changes will be called \textbf{variance-based test} and abbreviated as \textbf{Var} in the following sections. {\color[RGB]{0, 0, 0}Note that opposed to the GMD test, we do not need to assume that $s_1, s_2 > 0.5$, but we require the existence of higher moments.}

\subsection{Notation}
Define 
\begin{align*}
    \mathcal{I}_{h,k} := \{(h - 1)l_n + 1, ..., hl_n\} \times \{(k-1)l_m + 1, ..., kl_m\}, \quad h = 1, ..., b_n; \, k = 1, ..., b_m,
\end{align*}
as the set of indices of observations that fall into block $(h, k)$. Denote with
$$\hat{\nu}_{h,k} := \hat{\nu}_{h,k}^{(n,m)} := \frac{1}{l_n l_m} \underset{i,j \, \in \, \mathcal{I}_{h,k}}{\,\sum \sum\,} Y_{i,j} \overset{\text{H}_0}{=} \hat{\mu}_{h,k} - \mu$$
the block means of the random errors $(Y_{i,j})$ resp. the block means $\hat{\mu}_{h,k}^{(n,m)}$ if the {\color[RGB]{0, 0, 0} mean function $\mu$ was constantly equal to 0}. Under \hyp, where the location function $\mu$ is constant, {\color[RGB]{0, 0, 0}$\mu(x,y)\equiv\mu$,} this is equal to the mean of the observations in the block minus {\color[RGB]{0, 0, 0}the constant mean $\mu$ }. A tilde instead of a hat on an estimator indicates that the value was multiplied with the square root of the number of elements summed up, i.e. $\tilde{\mu}_{h,k} = \sqrt{l_n l_m} \hat{\mu}_{h,k}$ and $\nutilde = \sqrt{l_n l_m} \hat{\nu}_{h,k}.$ 

\subsection{Motivation}

{\color[RGB]{0, 0, 0}In a situation with centered data with a known variance $\sigma^2$, we could make practical use of the following proposition, which will also be helpful otherwise:}
\begin{proposition}\label{prop-CLT}
    Let Assumption~\ref{as-blockSize} be fulfilled and assume that $\ew(|Y_{i,j}|^{2+\delta}) < \infty$ for some $\delta > 0$. Then it holds under the hypothesis
    $$\sqrt{l_n l_m} \,\frac{\nuhat^{(n,m)}}{\sigma} \; \overset{\mathcal{D}}{\longrightarrow} \;\mathcal{N}(0, 1) \quad \forall 1 \leq h \leq b_n, \, 1 \leq k \leq b_m.$$
\end{proposition}
This convergence holds according to the central limit theorem for double arrays, see \mycitep{serfling}{31-32}. A more detailed proof is given in the appendix.

\begin{remark}
    According to the Continuous Mapping Theorem, 
    $$l_n l_m \frac{\nuhat^2}{\sigma^2} \overset{\mathcal{D}}{\longrightarrow} \mathcal{N}(0, 1)^2 = \chi^2_1 \quad \text{ and } \quad \frac{l_n l_m}{\sigma^2} \sum_{h = 1}^{a} \sum_{k = 1}^{b} \nuhat^2 \overset{\mathcal{D}}{\longrightarrow} \chi^2_{ab}$$
    for any $a \in \{1, ..., b_n\}$ and $b \in \{1, ..., b_m\}$. 
\end{remark}

\subsection{Proof of Theorem~\ref{th-LyapunovCLTfull}}

We are in the situation that the hypothesis holds and therefore the location function $\mu$ is constant. Since we cannot assume $\mu$ to be constantly equal to 0, we need to center the block means to derive a sensible CLT. Consider therefore the block means centered by the overall arithmetic mean:    
    \begin{align*}
        \sqrt{l_n l_m} \left(\muhat^{(n,m)} - \bar{X}\right) 
        &= \sqrt{l_n l_m} \left[ \frac{1}{l_n l_m} \doublesumBlock X_{i,j} - \bar{X}\right] \\
        &= \sqrt{l_n l_m} \left[ \frac{1}{l_n l_m} \doublesumBlock (\mu + Y_{i,j}) - (\mu + \bar{Y}) \right] \\
        &= \sqrt{l_n l_m} \left[ \frac{1}{l_n l_m} \doublesumBlock Y_{i,j} - \bar{Y}\right]\\
        &= \sqrt{l_n l_m}\left(\nuhat^{(n,m)} - \bar{Y}\right) \\
        &= \nutilde^{(n,m)} - \sqrt{l_n l_m}\bar{Y}, \quad h = 1, ..., b_n, \, k = 1, ..., b_m.
    \end{align*}
    Analogously, for the weighted sum of the squared values we can write 
    \begin{align*}
        \frac{l_n l_m}{\sqrt{b_n b_m}} \sum_{h = 1}^{b_n} \sum_{k = 1}^{b_m} \left(\muhat - \bar{X}\right)^2
        &= \frac{l_n l_m}{\sqrt{b_n b_m}} \sum_{h = 1}^{b_n} \sum_{k = 1}^{b_m} \left(\nuhat - \bar{Y}\right)^2 \\
        &= \frac{l_n l_m}{\sqrt{b_n b_m}} \left[ \sum_{h = 1}^{b_n} \sum_{k = 1}^{b_m} \nuhat^2 - b_n b_m \bar{Y}^2\right]\\
        &= \frac{1}{\sqrt{b_n b_m}} \left[\sum_{h = 1}^{b_n} \sum_{k = 1}^{b_m} \nutilde^2 - nm \bar{Y}^2 \right]
    \end{align*}
    since $\bar{Y} = \frac{1}{b_n b_m} \sum_{h = 1}^{b_n} \sum_{k = 1}^{b_m} \nuhat$. For the overall mean, it holds $\ew(\bar{Y}^2) = \frac{\sigma^2}{nm}$.

We can rewrite term~(\ref{eq-LyapunovCLTfull}) in Theorem~\ref{th-LyapunovCLTfull} as
\begin{align}
    (\ref{eq-LyapunovCLTfull}) =\;& \frac{1}{\sqrt{2}} \left(\frac{l_n l_m}{\hat{\sigma}^2 \sqrt{b_n b_m}} \sum_{h = 1}^{b_n} \sum_{k = 1}^{b_m} \nuhat^2 - \sqrt{b_n b_m} \right) \label{eq-LyapunovCLT(1)}\\
    &- \frac{1}{\sqrt{2}}\left(\frac{l_n l_m\sqrt{b_n b_m}}{\hat{\sigma}^2} \bar{Y}^2 - \frac{1}{\sqrt{b_n b_m}}\right). \label{eq-LyapunovCLT(2)}
\end{align}
The first term of the above difference contains the variance estimation on the block means and we can show that the Central Limit Theorem holds. The second term contains the centering by the arithmetic mean and we will show that it is asymptotically negligible. Before that, we introduce the following lemma which we will use later on.

\begin{lemma}\label{lemm-4thM}
Let $\ew(|Y_{1,1}|^{4 + \delta}) < \infty$ for $\delta > 0$. Denote with $\kappa_{n,m}^{(4)}$ the fourth standardized moment, {\color[RGB]{0, 0, 0}i.e. the kurtosis}, of the distribution of $\frac{1}{\sqrt{nm}} \sum_{i = 1}^n \sum_{j = 1}^m Y_{i,j}$. Then
    \begin{align*}
        \kappa_{n,m}^{(4)} = \ew\left[\frac{\left(\sum_{i = 1}^n \sum_{j = 1}^m Y_{i,j}\right)^4}{(nm)^2 \sigma^4}\right] = \ew\left[\frac{(nm)^2 \bar{Y}_{n,m}^4}{\sigma^4}\right] \longrightarrow 3 \quad \text{as } nm \rightarrow \infty.
    \end{align*}
\end{lemma}
The detailed proof is given in the appendix.

\begin{theorem}\label{th-LyapunovCLT-withoutCentering}
    Let {\color[RGB]{0, 0, 0}Assumption~\ref{as-blockSize}} be fulfilled and assume that $\ew(|Y_{i,j}|^{4 + \varepsilon}) < \infty \; \forall 1 \leq i \leq n, \, 1 \leq j \leq m$ and for $\varepsilon > 0$. Then it holds under the hypothesis
    $$\frac{1}{\sqrt{\kappa_{l_n, l_m}^{(4)} - 1}}\left[\frac{1}{\sigma^2\sqrt{b_n b_m}} \sum_{h = 1}^{b_n} \sum_{k = 1}^{b_m} \left(\nutilde^{(n,m)}\right)^2 - \sqrt{b_n b_m}\right] \overset{\mathcal{D}}{\longrightarrow}\mathcal{N}(0, 1) \, \text{ as } \, nm \rightarrow \infty$$
\end{theorem}

\begin{proof}
    $$\left(\nuhat^{(n,m)}\right) = \left(\nuhat^{(n,m)}, \, h = 1, ..., b_n, \, k = 1, ..., b_m, \, n,m \in \mathbb{N}\right)$$ is a double array, so $\left(\left(\nuhat^{(n,m)}\right)^2\right)$ is one, too. We notice that since the $Y_{i,j}$s are independent, the weighted, centered and squared non-overlapping block means $\left(\nuhat^2\right)$ are independent as well. Following from the independence of the $Y_{i,j}$, we get that
    $$\ew\left(l_n l_m \nuhat^2\right) = \ew\left(\nutilde^2\right) = \Var(\nutilde) + \underbrace{\ew(\nutilde)^2}_{= 0}\\
    = \frac{1}{l_n l_m} \doublesumBlock \Var(Y_{i,j}) = \sigma^2$$
    and 
    \begin{align*}
        \Var(l_n l_m \nuhat^2) &= \ew(\nutilde^4) - \ew(\nutilde^2)^2 =   \ew\left(\left(\frac{1}{\sqrt{l_n l_m}} \doublesumBlock Y_{i,j}\right)^4\right) - \sigma^4 = \sigma^4 (\kappa_{l_n, l_m}^{(4)} - 1) 
    \end{align*}
    where $\kappa_{l_n, l_m}^{(4)}$ is the fourth standardized moment of the distribution of $\nutilde^4$.
    {\color[RGB]{0, 0, 0} Using the inequality of \citet{marcinkiewicz1937fonctions}, see Theorem~2 in \citet{chow_teicher_1997}, Section~10.3}, we can show that the variance given in the last equation is indeed bounded. According to the Marcinkiewicz-Zygmund (M-Z) inequality, 
    $$\Var(\nutilde^2) = \sigma^4 (\kappa_{l_n, l_m}^{(4)} - 1) = \frac{1}{(l_n l_m)^2} \ew\left[ \left( \doublesumBlock Y_{i,j}\right)^4\right] - \sigma^4 < C - \sigma^4 < \infty$$
    for all $h,k$, for a $C > 0$ as long as the $Y_{i,j}$ are i.i.d. and $\ew(|Y_{1,1}|^4) < \infty$, as assumed. \\
    A detailed verification of Lyapunov's condition is given in the appendix. According to the CLT for triangular arrays, Theorem~\ref{th-LyapunovCLT-withoutCentering} holds.
    \end{proof}

As the following proposition shows, the terms $\kappa^{(4)}_{l_n, l_m}$ and $\sigma^2$ in Theorem~\ref{th-LyapunovCLT-withoutCentering} can be replaced by the corresponding limit value resp. estimator.

\begin{proposition} \label{prop-estimateVariance}
Let {\color[RGB]{0, 0, 0}Assumption~\ref{as-blockSize}} be fulfilled and let $\hat{\sigma}^2$ be a consistent estimator for $\sigma^2$ that converges in $\mathcal{O}\left(\frac{1}{\sqrt{nm}}\right)$. Then it holds under the hypothesis
    $$\frac{1}{\sqrt{2}} \left(\frac{l_n l_m}{\hat{\sigma}^2 \sqrt{b_n b_m}} \sum_{h = 1}^{b_n} \sum_{k = 1}^{b_m} \nuhat^2 - \sqrt{b_n b_m} \right) \overset{\mathcal{D}}{\longrightarrow} \mathcal{N}(0, 1) \, \text{ as } \, nm \rightarrow \infty$$
\end{proposition}

\begin{proof}
\begin{align*}
    &\frac{1}{\sqrt{2}} \left(\frac{l_n l_m}{\hat{\sigma}^2 \sqrt{b_n b_m}} \sum_{h = 1}^{b_n} \sum_{k = 1}^{b_m} \nuhat^2 - \sqrt{b_n b_m} \right) \\
    &= \underbrace{\frac{1}{\sqrt{\kappa^{(4)}_{l_n,l_m} - 1}} \left(\frac{l_n l_m}{\sigma^2 \sqrt{b_n b_m}} \sum_{h = 1}^{b_n} \sum_{k = 1}^{b_m} \nuhat^2 - \sqrt{b_n b_m} \right)}_{\overset{\mathcal{D}}{\longrightarrow} \; \mathcal{N}(0, 1) \; \text{according to Theorem~\ref{th-LyapunovCLT-withoutCentering}}} \frac{\sigma^2}{\hat{\sigma}^2 } \sqrt{\frac{\kappa^{(4)}_{l_n,l_m} - 1}{2}} + \sqrt{\frac{b_n b_m}{2}}\left(\frac{\sigma^2}{\hat{\sigma}^2 } - 1\right)
\end{align*}
    $\frac{\hat{\sigma}^2 }{\sigma^2} \rightarrow 1$ since $\hat{\sigma}^2$ is a consistent estimator for $\sigma^2$ and as per the CMT, $\frac{\sigma^2}{\hat{\sigma}^2 } \rightarrow 1$, too. $\sqrt{\frac{\kappa^{(4)}_{l_n,l_m} - 1}{2}} \rightarrow \sqrt{\frac{2}{2}} = 1$ holds according to Lemma~\ref{lemm-4thM}. Therefore, by  application of Slutzky's lemma, we get the convergence in distribution of the first summand to a standard normal distribution. \\
    For the second summand, we know that $\sigma^2 - \hat{\sigma}^2$ converges to 0 in $\mathcal{O}\left(\frac{1}{\sqrt{nm}}\right)$, hence
    $$\sqrt{\frac{b_n b_m}{2}} \frac{\sigma^2 - \hat{\sigma}^2}{\sigma^2} = \sqrt{\frac{N^{1 - s}}{2}} \frac{\sigma^2 - \hat{\sigma}^2}{\sigma^2} \, \longrightarrow 0 \quad (nm \rightarrow \infty)$$
    as long as $s < 1$. And with the same argument as before, 
    $$\sqrt{\frac{b_n b_m}{2}} \frac{\sigma^2 - \hat{\sigma}^2}{\sigma^2} \frac{\sigma^2}{\hat{\sigma}^2} = \sqrt{\frac{b_n b_m}{2}} \left(\frac{\sigma^2}{\hat{\sigma}^2} - 1 \right) \; \longrightarrow 0 \quad(nm \rightarrow \infty).$$
    Another application of Slutzky's lemma proves the proposition.
\end{proof}

\begin{remark}
    The empirical variance $\frac{1}{n-1}\sum_{i = 1}^n\sum_{j = 1}^m (X_{i,j} - \bar{X})^2$ over all observations of the random field is a consistent estimator for $\sigma^2$ under \hyp. As long as $\ew(|Y|^4) < \infty$, which is assumed throughout this paper, the empirical variance converges to $\sigma^2$ in $\mathcal{O}\left(\frac{1}{{\color[RGB]{0, 0, 0}\sqrt{nm}}}\right)$ (\citeI{serfling}, p. 192).
\end{remark}

\begin{proposition}\label{prop-expMeanconvergence}
    Let $\hat{\sigma}^2$ be a (weakly) consistent estimator for $\sigma^2$. Then
    $${\color[RGB]{0, 0, 0}\frac{l_n l_m \sqrt{b_n b_m}}{\hat{\sigma}^2} \bar{Y}^2} \overset{P}{\longrightarrow} 0 \quad \text{as } {\color[RGB]{0, 0, 0}nm} \rightarrow \infty.$$
\end{proposition}
\begin{proof}
{\color[RGB]{0, 0, 0}Re-writing the term leads to 
$$\frac{l_n l_m \sqrt{b_n b_m}}{\hat{\sigma}^2} \bar{Y}^2 = \underbrace{\frac{\sigma^2}{\hat{\sigma}^2}}_{\overset{P}{\rightarrow} 1} \frac{l_n l_m \sqrt{b_n b_m}}{\sigma^2} \bar{Y}^2$$}
    Application of Markov's inequality yields
    \begin{align*}\color[RGB]{0, 0, 0}
        \P\left(\left|\frac{l_n l_m \sqrt{b_n b_m}}{\sigma^2} \bar{Y}^2  \right| > \varepsilon \right) 
        &\color[RGB]{0, 0, 0}\leq  \frac{1}{\,\varepsilon\,} \,\ew\left(\left|  \frac{l_n l_m \sqrt{b_n b_m}}{\sigma^2} \bar{Y}^2  \right| \right) \\
        &\color[RGB]{0, 0, 0}= \frac{l_n l_m \sqrt{b_n b_m}}{\varepsilon\sigma^2} \ew\left(\bar{Y}^2 \right)  = \frac{l_n l_m \sqrt{b_n b_m}\sigma^2}{\varepsilon \sigma^2 nm}\\
        &\color[RGB]{0, 0, 0}= \frac{1}{\varepsilon \sqrt{b_n b_m}}  \; \longrightarrow 0 \; (nm \rightarrow \infty).
    \end{align*}
\end{proof}

Rewriting term~(\ref{eq-LyapunovCLTfull}) in Theorem~\ref{th-LyapunovCLTfull} leads us to
\begin{align*}
    (\ref{eq-LyapunovCLTfull}) = \frac{1}{\sqrt{2}} \left(\frac{l_n l_m}{\hat{\sigma}^2 \sqrt{b_n b_m}} \sum_{h = 1}^{b_n} \sum_{k = 1}^{b_m} \nuhat^2 - \sqrt{b_n b_m} \right) - \frac{1}{\sqrt{2}}\left(\frac{l_n l_m\sqrt{b_n b_m}}{\hat{\sigma}^2} \bar{Y}^2 - \frac{1}{\sqrt{b_n b_m}}\right)
\end{align*}
The convergence of the first term of the difference to a standard normal distribution {was proven in \color[RGB]{0, 0, 0}Proposition~\ref{prop-estimateVariance}}. With Proposition~\ref{prop-expMeanconvergence} and a second application of Slutzky's lemma, Theorem~\ref{th-LyapunovCLTfull} is proven. 
 \begin{remark}
     Note that we included a "$+1$" in~(\ref{eq-LyapunovCLTfull}) as a {\color[RGB]{0, 0, 0}correction for} the effect of {\color[RGB]{0, 0, 0} estimating} $\bar{X}^2$. This is negligible for the asymptotic result. However, since $b_n$ and $b_m$ tend to be rather small even in moderate sample sizes, simulations have shown that its inclusion is beneficial to the power of the test. 
 \end{remark}

\subsection{Consistency under the Alternative}

\begin{theorem}
Let {\color[RGB]{0, 0, 0}Assumptions~\ref{as-locFun} and \ref{as-blockSize}} be fulfilled and assume that $\ew(|Y_{1,1}|^{4 + \delta}) < \infty$ for $\delta > 0$. Then both under the hypothesis and the alternative it holds that
    \begin{align*} 
    T(n, m) &:= \frac{1}{b_n b_m} \sum_{h = 1}^{b_n} \sum_{k = 1}^{b_m} \left(\muhat - \bar{X}\right)^2 \overset{L^2}{\longrightarrow} \int_0^1 \int_0^1 \left(\mu(x, y) - \int_0^1 \int_0^1 \mu(u, v)\, \dd u \dd v \right)^2 \dd x \dd y \quad \text{ as } \, n,m \rightarrow \infty.
    \end{align*}
\end{theorem}
The double integral above determines the variability of the function $\mu$. It is well known that it is 0 if and only if $\mu$ is constant except for Lebesgue-negligible sets. Due to this, it follows that the test is consistent against all other non-constant mean functions.

\begin{proof}
In this proof, we use the {\color[RGB]{0, 0, 0}short-hand} notation {\color[RGB]{0, 0, 0}$\mu_{i,j} := \mu\left(\frac{i}{n}, \frac{j}{m}\right)$}. Denote with $\mu_{sup} := \sup_{x, y \in [0,1]} |\mu(x, y)|$ the maximal absolute value $\mu$ takes {\color[RGB]{0, 0, 0}and with $\bar{\mu}_{n,m} = \bar{\mu} = \frac{1}{nm} \sum_{i = 1}^n\sum_{j = 1}^m \mu_{i, j}$ the mean of the location function over all sampled locations.} First, we show that
$$T(n, m) \approx \frac{1}{b_n b_m} \sum_{h = 1}^{b_n} \sum_{k = 1}^{b_m} \left(\mu_{h l_n, k l_m} - \bar{\mu}_{\cdot l_n, \cdot l_m}\right)^2,$$
i.e., the arguments in the variance function in $T$ can be replaced by deterministic ones. The $\mu_{hl_n, kl_m} = \mu\left(\frac{h}{b_n}, \frac{k}{b_m}\right)$ can be seen as "block representatives" opposed to the block means $\hat{\mu}_{hk}$ in the original function $T$. Then $\bar{\mu}_{\cdot l_n, \cdot l_m} := \frac{1}{b_n b_m} \sum_{h = 1}^{b_n} \sum_{k = 1}^{b_m} \mu_{hl_n, kl_m}$ is the mean over all block representatives. Since $X_{i,j} = \mu_{i,j} + Y_{i,j}$ it follows that 
$$\muhat = \frac{1}{l_n l_m} \doublesumBlock X_{i,j} = \frac{1}{l_n l_m} \doublesumBlock ({\color[RGB]{0, 0, 0}\mu_{i, j}} + Y_{i,j}) = \bar{\mu}_{h,k} + {\color[RGB]{0, 0, 0}\bar{\nu}_{h,k}}.$$
Using this, $T(n, m)$ can be written as
\begin{align*}
T(n, m) &= \frac{1}{b_n b_m} \sum_{h = 1}^{b_n} \sum_{k = 1}^{b_m} \left(\muhat - \bar{X}\right)^2 \\
&= \frac{1}{b_n b_m}  \sum_{h = 1}^{b_n} \sum_{k = 1}^{b_m} \left(\bar{\mu}_{h,k} - \bar{\mu} + {\color[RGB]{0, 0, 0}\bar{\nu}_{h,k}} - \bar{Y} \right)^2 \\
&= \frac{1}{b_n b_m}  \sum_{h = 1}^{b_n} \sum_{k = 1}^{b_m} \left( (\bar{\mu}_{h,k} - \bar{\mu})^2 + 2(\bar{\mu}_{h,k} - \bar{\mu})({\color[RGB]{0, 0, 0}\bar{\nu}_{h,k}} - \bar{Y}) +  ({\color[RGB]{0, 0, 0}\bar{\nu}_{h,k}} - \bar{Y})^2 \right).
\end{align*}
Consequently,
\begin{align*}
    &\left| T(n, m) -  \frac{1}{b_n b_m} \sum_{h = 1}^{b_n} \sum_{k = 1}^{b_m} \left(\mu_{h l_n, k l_m} - \bar{\mu}_{\cdot l_n, \cdot l_m}\right)^2\right| \\
    &=\left| \frac{1}{b_n b_m}  \sum_{h = 1}^{b_n} \sum_{k = 1}^{b_m} \left( (\bar{\mu}_{h,k} - \bar{\mu})^2 -  \left(\mu_{h l_n, k l_m} - \bar{\mu}_{\cdot l_n, \cdot l_m}\right)^2 + 2(\bar{\mu}_{h,k} - \bar{\mu})({\color[RGB]{0, 0, 0}\bar{\nu}_{h,k}} - \bar{Y}) +  ({\color[RGB]{0, 0, 0}\bar{\nu}_{h,k}} - \bar{Y})^2 \right)\right|\\
    &\leq \underbrace{\left| \frac{1}{b_n b_m}  \sum_{h = 1}^{b_n} \sum_{k = 1}^{b_m} \left( (\bar{\mu}_{h,k} - \bar{\mu})^2 -  \left(\mu_{h l_n, k l_m} - \bar{\mu}_{\cdot l_n, \cdot l_m}\right)^2\right) \right| }_{\text{(I)}}
    + \underbrace{\left| \frac{1}{b_n b_m} \sum_{h = 1}^{b_n} \sum_{k = 1}^{b_m} ({\color[RGB]{0, 0, 0}\bar{\nu}_{h,k}} - \bar{Y})^2 \right|}_{\text{(II)}}\\
    &\quad\quad + \underbrace{\left|\frac{2}{b_n b_m}  \sum_{h = 1}^{b_n} \sum_{k = 1}^{b_m} (\bar{\mu}_{h,k} - \bar{\mu})({\color[RGB]{0, 0, 0}\bar{\nu}_{h,k}} - \bar{Y}) \right|}_{\text{(III)}}.
\end{align*}
We will now treat each of the three terms individually.\\

(I):
Turning to the first term, we have to make sure to include possible change regions $\mathcal{C}_1, ..., \mathcal{C}_r$ into the calculation. As long as Assumption~\ref{as-locFun} holds, we can divide the blocks into two sets. Denote by $\mathcal{R}^{(1)}$ the set of indices $1 \leq h \leq b_n, \; 1 \leq k \leq b_m$ where the block with index $(h, k)$ intersects with exactly one region (either with {\color[RGB]{0, 0, 0}or }without change). Accordingly, $\mathcal{R}^{(+)}$ is the set of those indices where the corresponding blocks intersect with more than one region. {\color[RGB]{0, 0, 0}The number of blocks in $\mathcal{R}^{(+)}$ is $s^{(n,m)}=\mathcal{O}(b_n+b_m)$: We divide our plane $[0, 1]^2$ into $b_n \times b_m$ blocks, each of size $\frac{1}{b_n} \times \frac{1}{b_m}$. We get equidistant grid lines with distance $\frac{1}{b_n}$ (horizontally) resp. $\frac{1}{b_m}$ (vertically). Since the boundary of each $\mathcal{C}_i$ has a finite length $\mathcal{L}_i$, it can intersect with at most $\left\lceil \frac{\mathcal{L}_i}{\frac{1}{b_n}} \right\rceil = \left\lceil \mathcal{L}_i b_n \right\rceil$ horizontal grid lines. Therefore, a straight line parallel to the $x$-axis can enter at most $\left\lceil \mathcal{L}_i b_n \right\rceil + 1$ blocks. Analogously, a straight line parallel to the $y$-axis can enter at most $\left\lceil \mathcal{L}_i b_m \right\rceil + 1$ blocks. The number of blocks that the border of $\mathcal{C}_i$ intersects with is therefore bounded by $\mathcal{O}\left(\mathcal{L}_i(b_n + b_m)\right)=\mathcal{O}(b_n+b_m)$. Hence the number of blocks that intersect with more than one region fulfills $s^{(n,m)} = \mathcal{O}(b_n + b_m)$.
}

Now term (I) can be split up as follows:
\begin{align*}
    \left| \frac{1}{b_n b_m}  \sum_{h = 1}^{b_n} \sum_{k = 1}^{b_m} \left( (\bar{\mu}_{h,k} - \bar{\mu})^2 -  \left(\mu_{h l_n, k l_m} - \bar{\mu}_{\cdot l_n, \cdot l_m}\right)^2\right) \right| 
    &= \left| \frac{1}{b_n b_m}  \sum_{h = 1}^{b_n} \sum_{k = 1}^{b_m} \left((\bar{\mu}_{h,k}^2 - \bar{\mu}^2) -  (\mu_{h l_n, k l_m}^2 - \bar{\mu}_{\cdot l_n, \cdot l_m}^2) \right) \right| \\
    &=\left| \frac{1}{b_n b_m}  \sum_{h = 1}^{b_n} \sum_{k = 1}^{b_m}  \left(\bar{\mu}^2_{h,k} - \mu_{h l_n, k l_m}^2  - \bar{\mu}^2 + \bar{\mu}^2_{\cdot l_n, \cdot l_m} \right)\right|\\
    &\leq \frac{1}{b_n b_m}  \sum_{h = 1}^{b_n} \sum_{k = 1}^{b_m} \left| \bar{\mu}^2_{h,k} - \mu_{h l_n, k l_m}^2\right| + \left|\bar{\mu}^2 - \bar{\mu}^2_{\cdot l_n, \cdot l_m} \right|
\end{align*}
{\color[RGB]{0, 0, 0} Define with 
$$\widetilde{\mathcal{R}}^{(1)} = \bigcup_{(h,k) \in \mathcal{R}^{(1)}} \left[\frac{h-1}{b_n}, \frac{h}{b_n}\right] \times \left[\frac{k-1}{b_m}, \frac{k}{b_m}\right]$$
the continuous extension of $\mathcal{R}^{(1)}$.}
To find an upper bound for term (I), we can make use of the continuity of the function $\mu$ over {\color[RGB]{0, 0, 0}$\widetilde{\mathcal{R}}^{(1)}$}. For some arbitrary $(i,j) \in \mathcal{I}_{hk}$ it follows
$$d\left(\left(\frac{i}{n}, \frac{j}{m}\right), \left(\frac{h}{b_n}, \frac{k}{b_m}\right)\right) \leq d\left(\left(\frac{(h - 1)l_n + 1}{n}, \frac{(k - 1)l_m + 1}{m}\right), \left(\frac{h}{b_n}, \frac{k}{b_m}\right)\right) \rightarrow 0$$
since both $\left|\frac{(h - 1)l_n + 1}{n} - \frac{h}{b_n}\right| \rightarrow 0$ and $\left| \frac{(k - 1)l_m + 1}{m} - \frac{k}{b_m}\right| \rightarrow 0 \; \forall h, k$ as $n,m \rightarrow \infty$. As $\mu'$ is continuous and therefore uniformly continuous over $[0, 1]$, it holds that $\forall \delta > 0 \; \exists \tilde{n}, \tilde{m} \in \mathbb{N}$ such that 
\begin{align*}
    \max_{(h, k) \in \mathcal{R}^{(1)}} \sup_{(i,j) \in \mathcal{I}_{h,k}} \left| \mu_{i,j} - \mu_{hl_n, kl_m}\right| = |\mu'_{i,j} - \mu'_{hl_n, kl_m}|  < \delta \quad \forall n > \tilde{n}, m > \tilde{m}.
\end{align*} 
This inequality still applies if we replace $\mu_{i,j}$ with the corresponding arithmetic block mean $\bar{\mu}_{h,k}.$ From the uniform continuity of the function $g: x \mapsto x^2$ over $[0, 1]^2$ it follows that $\forall \varepsilon_1 > 0$ $\exists \tilde{n},\tilde{m} \in \mathbb{N}$ with 
$$\max_{(h, k)\in \mathcal{R}^{(1)}} \left| \bar{\mu}^2_{h,k} - \mu^2_{hl_n, kl_m}\right| < \varepsilon_1  \quad \forall n > \tilde{n}, m > \tilde{m}.$$
On the other hand, for all $(h, k) \in \mathcal{R}^{(+)}$ we get $\left| \bar{\mu}^2_{h,k} - \mu^2_{hl_n, kl_m}\right| \leq 2\mu_{sup}^2$. Consequently,
\begin{align} \label{eq-ContFun1}
    \frac{1}{b_n b_m} \sum_{h = 1}^{b_n} \sum_{k = 1}^{b_m} \left|\bar{\mu}^2_{h,k} - \mu^2_{hl_n, kl_m} \right| < \frac{2s^{(n,m)}}{b_n b_m} \mu^2_{sup} + \frac{b_n b_m - s^{(n,m)}}{b_n b_m} \varepsilon_1 < \varepsilon_2
\end{align}
for $\varepsilon_2 > 0$ and  $n, m$ large enough, since we have $b_n,b_m\to \infty$. \\
Turning to the second term, for all $(h, k) \in \mathcal{R}^{(1)}$ we have $\left| \bar{\mu}_{h,k} - \mu_{hl_n, kl_m}\right| < \delta$ for $n, m$ sufficiently large as above. For $(h, k) \in \mathcal{R}^{(+)}$ we get that $\left| \bar{\mu}_{h,k} - \mu_{hl_n, kl_m}\right| \leq 2\mu_{sup}$. In total, this yields 
$$\left|\bar{\mu} - \bar{\mu}_{\cdot l_n, \cdot l_m} \right| \leq \frac{1}{b_n b_m} \sum_{h = 1}^{b_n} \sum_{k = 1}^{b_m} \left|\bar{\mu}_{h,k} - \mu_{hl_n, kl_m} \right| < \frac{2s^{(n,m)}}{b_n b_m} \mu_{sup} + \frac{b_n b_m - s^{(n,m)}}{b_n b_m} \delta < \delta$$
for $n, m$ large enough and by continuity we derive that $\forall \varepsilon_3 > 0 \; \exists \tilde{n}, \tilde{m}\in \mathbb{N}$ such that
\begin{align}\label{eq-ContFun2}
    \left|\bar{\mu}^2 - \bar{\mu}^2_{\cdot l_n, \cdot l_m} \right| < \varepsilon_3  \quad \forall n > \tilde{n}, m > \tilde{m}.
\end{align}
Combining (\ref{eq-ContFun1}) and (\ref{eq-ContFun2}), we get 
\begin{align*}
     \frac{1}{b_n b_m}  \sum_{h = 1}^{b_n} \sum_{k = 1}^{b_m} \left| \bar{\mu}^2_{h,k} - \mu_{h l_n, k l_m}^2\right| + \left|\bar{\mu}^2 - \bar{\mu}^2_{\cdot l_n, \cdot l_m} \right| < \varepsilon_2 + \varepsilon_3 < \varepsilon
\end{align*}
for {\color[RGB]{0, 0, 0}$n,m$} large enough.
\\

(II):
{\color[RGB]{0, 0, 0}
\begin{align*}
    \ew\left[\left|\frac{1}{b_n b_m} \sum_{h = 1}^{b_n} \sum_{k = 1}^{b_m} ({\color[RGB]{0, 0, 0}\bar{\nu}_{h,k}} - \bar{Y})^2 \right|^2\right] 
    &= \ew\left[\left|\frac{1}{b_n b_m} \sum_{h = 1}^{b_n} \sum_{k = 1}^{b_m} \left(\frac{1}{nm - l_n l_m} \underset{(i,j) \notin \mathcal{I}_{h,k} }{\sum_{i = 1}^n \sum_{j = 1}^m} Y_{i,j}\right)^2 \right|^2\right] \\
    &\hspace{-2mm}\overset{c_r\text{-ineq.}}{\leq} \frac{2 }{(b_n b_m (nm - l_n l_m))^2} \sum_{h = 1}^{b_n} \sum_{k = 1}^{b_m} \ew\left[\left|\underset{(i,j) \notin \mathcal{I}_{h,k} }{\sum_{i = 1}^n \sum_{j = 1}^m} \frac{Y_{i,j}}{ \sqrt{nm - l_n l_m}}\right|^4\right] \\
    &= \frac{2 \sigma^4}{b_n b_m (nm - l_n l_m)^2} \underbrace{\ew\left[\left|\underset{(i,j) \notin \mathcal{I}_{h,k} }{\sum_{i = 1}^n \sum_{j = 1}^m} \frac{Y_{i,j}}{\sigma \sqrt{nm - l_n l_m}}\right|^4\right]}_{\rightarrow 3 \; (nm \rightarrow \infty)} \; \longrightarrow 0 \; \text{as } nm \rightarrow \infty
\end{align*}
Thus, $\frac{1}{b_n b_m} \sum_{h = 1}^{b_n} \sum_{k = 1}^{b_m} \left({\color[RGB]{0, 0, 0}\bar{\nu}_{h,k}} - \bar{Y}\right)^2 \overset{L^2}{\longrightarrow} 0$ as $nm \rightarrow \infty$.}


(III): 
Using the Cauchy-Schwarz-inequality, we get
\begin{align*}\addtolength{\jot}{1em}
    &\left|\frac{2}{b_n b_m} \sum_{h = 1}^{b_n} \sum_{k = 1}^{b_m} (\bar{\mu}_{h,k} - \bar{\mu})({\color[RGB]{0, 0, 0}\bar{\nu}_{h,k}} - \bar{Y}) \right| 
    \leq \left|  \sqrt{\frac{2}{b_n b_m} \sum_{h = 1}^{b_n} \sum_{k = 1}^{b_m} (\bar{\mu}_{h,k} - \bar{\mu})^2} \cdot \sqrt{\frac{2}{b_n b_m}\sum_{h = 1}^{b_n} \sum_{k = 1}^{b_m} ({\color[RGB]{0, 0, 0}\bar{\nu}_{h,k}} - \bar{Y})^2} \right|.
\end{align*}
The first term on the right hand side is bounded since
\begin{align*}
    \sqrt{\frac{2}{b_n b_m} \sum_{h = 1}^{b_n} \sum_{k = 1}^{b_m} (\bar{\mu}_{h,k} - \bar{\mu})^2}
    \leq \sqrt{2}\left(\max(\mu) - \min(\mu)\right) = const.
\end{align*}
Analogously to (II), for the second term we get that 
\begin{align*} 
    \sqrt{\frac{2}{b_n b_m}\sum_{h = 1}^{b_n} \sum_{k = 1}^{b_m} ({\color[RGB]{0, 0, 0}\bar{\nu}_{h,k}} - \bar{Y})^2} \; \overset{L^2}{\longrightarrow} 0 \; (nm \rightarrow\infty).
\end{align*}
In total, we get that the term (III) converges to 0 as $nm$ goes to $\infty$.\\

It remains to show that 
$$\frac{1}{b_n b_m} \sum_{h = 1}^{b_n} \sum_{k = 1}^{b_m} \left(\mu_{h l_n, k l_m} - \bar{\mu}_{\cdot l_n, \cdot l_m}\right)^2 \rightarrow \int_0^1 \int_0^1 \left(\mu(x, y) - \int_0^1 \int_0^1 \mu(u, v) du dv\right)^2 dx dy.$$
By assumption, the discontinuities of the location function $\mu$ form a Lebesgue null set and $\mu$ is bounded on $[0, 1]^2$. Hence 
$$\bar{\mu}_{\cdot l_n, \cdot l_m} = \frac{1}{b_n b_m} \sum_{h = 1}^{b_n} \sum_{k = 1}^{b_m} \mu_{hl_n, kl_m} \rightarrow \int_0^1 \int_0^1 \mu(u, v) du dv$$
by Lebesgue's integrability criterion for multiple Riemann integrals (e.g. Theorem~14.5 in \citeI{Apostol1974}). Similarly, for an arbitrary constant $C \in \mathbb{R}$, $g_C(x, y) := (\mu(x, y) - C)^2$ has only discontinuities with Lebesgue-measure 0 and is bounded on $[0, 1]^2$. By Lebesgue's integrability criterion for multiple Riemann integrals, and by choosing $C~=~\int_0^1 \int_0^1 \mu(x, y) dx dy$, we get the desired convergence.
\end{proof}

{\color[RGB]{0, 0, 0}\begin{remark}
    Note that for the convergence of the Var test statistic $T^{Var}_{n,m}$ to the standard normal distribution, we only require the product $nm$ to tend to infinity. However, for the consistency under the alternative, both $n$ and $m$ need to increase to infinity, independent of each other.
\end{remark}}

\section{\color[RGB]{0, 0, 0}Extensions to dependent data} \label{sec-decorr}

{\color[RGB]{0, 0, 0}In case of observing values from a dependent random field, there are several possible modifications of our test statistic. The choice of a suitable modification depends on our knowledge of the dependence structure. Note that we assume the dependence between neighboring observations to be fixed and independent of the sample size, which is common practice in the change-point literature, see e.g. \citet{schmidt2024detecting} or \citet{kirch2025scan}. 

If our data are $M$-dependent, i.e., $Y_{i,j}$ and $Y_{i',j'}$ are independent if $\max\{|i-i'|,|j-j'|\}>M$, we can proceed as follows. We split our blocks consisting of $l_n\times l_m$ observations into smaller subblocks of size $\tilde{l}_n\times \tilde{l}_m$, where $\tilde{l}_n=l_n^r$ and $\tilde{l}_m=l_m^r$ for some $0<r<1$. In this way we get an increasing number of subblocks of increasing size for each block. Then we replace the observations in each subblock by the average of the observations in its upper left corner, removing the lower $M$ rows and the right $M$ columns of observations from it. The means of the reduced subblocks will be independent by construction under the assumptions of $M$-dependence, so that we can apply our test statistics to the spatial field of $\tilde{n}\times\tilde{m}=b_nl_n^{1-r}\times b_ml_m^{1-r}$ reduced subblock means. As we drop some rows and some columns of observations to achieve this independence, we will loose the information in a certain fraction of the observations. In case of $M=1$ and $\tilde{l}_n=\tilde{l}_m=3$ e.g., we will calculate the subblock means from 4 out of the 9 observations in each block. Note that the fraction of observations not used for the test statistic decreases from $5/9\approx 55.6\%$ if $\tilde{l}_n=\tilde{l}_m=3$ to $7/16\approx 43.8\%$ if $\tilde{l}_n=\tilde{l}_m=4$ and $9/25=36\%$ if $\tilde{l}_n=\tilde{l}_m=5$, so that the loss will become small if $m$ and $n$ are large. In the remainder of the paper, this method will be called "cut-off method".}\\

{\color[RGB]{0, 0, 0}Another, more widely  applicable} option to modify the tests for dependent data is to remove the correlation in the data beforehand. This can either be done by fitting a suitable model and working with the residuals, or, if the structure of the data is unknown or no suitable model exists, one can de-correlate the data using their sample autocovariance. This is a well-established approach, proposed, e.g., by \myCite{robbins2011mean} for detecting a shift in time series. According to their study, tests especially designed for dependent data only have slightly higher power than tests for independent data applied to one-step-ahead prediction residuals, but they come with substantial computational complexities. Other reasons to avoid tests adapted to dependence may be analytically unknown critical values or strict model assumptions that prevent the adapted test to be generalized to broader settings.  

The basis for our de-correlation algorithm is the assumption of stationarity of the data under the hypothesis. We need to estimate their autocovariances $\gamma(\bb{h})$ for all relevant lags $\bb{h} = (h_1, h_2)$. This is done by the empirical estimator
$$\hat{\gamma}_{\text{reg}}(h_1, h_2) =  \frac{1}{N}\sum_{i = 1}^{n-h_1} \sum_{j = 1}^{m-h_2} (X_{i,j} - \bar{X})(X_{i+h_1, j+h_2} - \bar{X}). $$
The autocovariances are estimated for all lags up to an upper bound $\bb{b}^{(n,m)} = (b_1^{(n)}, b_2^{(m)})$, all other autocovariances are set to 0. In this paper, $b_i^{(k)} = \left\lfloor 0.9 k^{(1/3)} \right\rfloor$ is used since it resembles the recommendations of \myCite{andrews1991heteroskedasticity} for kernel density estimation, and it showed good results in preliminary studies. 

Having obtained the estimated autocovariances, we order the data matrix $X$ into a vector $x = \text{vec}(X)$. Then, using all estimated $\hat{\gamma}(\bb{h})$, we construct the estimated autocovariance matrix $\hat{\Sigma}$ of the data vector $x$. To obtain the square root of the matrix $\hat{\Sigma}$, we perform a Cholesky decomposition, or the revised modified Cholesky decomposition (\citeI{schnabel1990new}) if the estimated autocovariance matrix is not positive-semidefinite. We invert this square root using the default \texttt{R} function \texttt{inv()} that is based on the LAPACK routine DGESV (\citeI{lapack99}). 
The de-correlation process is then performed as $y = \hat{\Sigma}^{-\frac{1}{2}} (x - \bar{x})$. Finally, we reorder $y$ back into a matrix $Y$ column-wise. Instead of regular autocovariances, one could also use a difference-based approach (see e.g. \citeI{tecuapetla2017autocovariance}). 

If the assumption of separability of the covariance function is justified, we can reduce the costs of decomposing and inverting $\hat{\Sigma}$ by estimating two smaller covariance matrices: one for the horizontal ($\hat{\Sigma}_1 \sim (m \times m)$) and one for the vertical ($\hat{\Sigma}_2 \sim (n \times n)$) direction. It holds that $\Sigma = \Sigma_1 \otimes \Sigma_2$. Accordingly, $\hat{\Sigma}_1$ and $\hat{\Sigma}_2$ can be estimated, decomposed, and inverted separately, reducing computation time from $\mathcal{O}(n^3m^3)$ to $\mathcal{O}(n^3 + m^3)$. It follows that $\hat{\Sigma}^{-\frac{1}{2}} = \hat{\Sigma}_1^{-\frac{1}{2}} \otimes \hat{\Sigma}_2^{-\frac{1}{2}}$.

 In the subsequent simulation study, we also investigated the behavior of the tests using such a difference-based estimator, but could not find any meaningful advantages to the regular one. Alternatively, if a suitable model for the data is known, one can fit that model and apply tests to its residuals for further analysis.

\section{Simulation study} \label{sec-sim}

In this section we analyze the finite sample behavior of the GMD and the VAR test under the hypothesis of a constant mean and under several alternatives. We do so using Monte Carlo simulations. The simulations are conducted using the software \texttt{R} (\citeI{R}, version 4.4.3) along with the packages \texttt{SChangeBlock} (\citeI{SChangeBlock}), \texttt{robcp} (\citeI{robcp}), and \texttt{ggplot2} (\citeI{ggplot2}).

\subsection{Setup}

We choose the dimension of a random field to be $n\times n$ such that $n = m$ and $N = n^2$. To build the blocks, we choose $s_n = s_m = s$ around 0.6 such that for the block length $l_n = [n^s]$ and the number of blocks per dimension $b_n = \lfloor n / l_n\rfloor$ it holds $l_n \cdot b_n = n$. Preliminary studies indicated that both tests work best if there neither are blocks at the edges of the random field that are smaller than the majority of the blocks, nor if such blocks are left out completely. For generating the noise $(Y_{i,j}: i, j \in  \{1, ..., n\})$ we consider three different distributions, namely the standard normal distribution $Y_{i,j} \sim \mathcal{N}(0, 1)$, the $t$-distribution with 3 degrees of freedom $Y_{i,j} \sim t_3$, and the $\chi^2_2$ distribution with 2 degrees of freedom, which equals the Exp(\nicefrac{1}{2}) distribution. Contrary to the requirements in Theorem~\ref{th-LyapunovCLTfull}, the $t_3$ distribution does not possess finite $(4 + \varepsilon)$-th absolute moments. 

To incorporate dependency, we use a symmetric Spatial Moving Average model of order $q$ (short: SMA$(q)$): 
$$Y_{i,j} = \sum_{k = -q}^q \sum_{l = -q}^q\theta_{k,l} \varepsilon_{k,l}.$$ 
The parameters $(\theta_{k,l})$ are chosen as $\theta_{k,l} = \left(\frac{\rho}{2}\right)^{|k - q - 1| + |l - q - 1|}$ with parameter $\rho$ for a pure SMA$(q)$ model. This yields $M = q + 1$-dependent data. Alternatively, we {\color[RGB]{0, 0, 0} simulate a SAR(1) field using
$$Y_{i,j} = \rho (Y_{i-1, j} + Y_{i, j-1}) + \varepsilon_{i,j}$$
with a constant parameter $\rho \in (-0.5, 0.5)$. We choose a burn-in period of 20, i.e. the first 20 simulated values in both directions are omitted.}
The distribution of the noise $(\varepsilon_{k,l})$ is chosen to be standard normal. For the simulations, we consider an SMA(1) and a SAR(1) model, each with parameters $\rho \in \{0.1, 0.2, 0.3\}$. 
{\color[RGB]{0, 0, 0} For both dependency structures, we investigate the behavior of the Var test if the simulated data is de-correlated. Since SMA(1) models as formulated above are 2-dependent, for this type of dependency we also investigate the effect of the cut-off method with subblocks of size $\tilde{l}_n = \tilde{l}_m = 4$. Even though 75\% of the data is omitted that way, we get the advantage of having more new data points as compared to a larger subblock size. In a preliminary study, we discovered that at least for small sample sizes, choosing $\tilde{l}_n = \tilde{l}_m = 4$ works best. As SAR(1) models are not $M$-dependent, using the cut-off method would not make sense or require a large value of $M$ (corresponding to a huge loss of information) to get a reasonable approximation. As opposed to SMA fields, fitting a SAR field is much easier. Therefore, we included the application of the Var test to the residuals of a SAR(1) model fitted by GLS, using the \texttt{R} package \texttt{spatialreg} (\citeI{spatialreg}).}

We test the hypothesis of a constant $\mu$ in $X_{i,j} = \mu(\nicefrac{i}{n}, \nicefrac{j}{m}) + Y_{i,j}, \, 1 \leq i,j \leq n$, against the alternative that $\mu$ changes across the field. For instance, $\mu$ could abruptly shift in some area or $\mu$ could steadily increase from one end of the field to the other. We will investigate the behavior of the tests on the following four different alternatives:\vspace{-3mm}\\
\begin{table}[ht]
    \centering
    \begin{tabular}{r l}
         $\mathbb{A}_1$: & $\mu(x, y) = \frac{H}{\sqrt{nm}} \cdot \mathbbm{1}_{\{0 \leq x \leq \frac{1}{l_n}\}}  \mathbbm{1}_{\{0 \leq y \leq \frac{1}{l_m}\}}$ \\
         $\mathbb{A}_2$: & $\mu(x, y) = \frac{H}{\sqrt{nm}} \cdot \mathbbm{1}_{\{0 \leq y \leq \frac{1}{2}\}}$ \\
         $\mathbb{A}_3$: & $\mu(x, y) = \frac{H}{\sqrt{nm}} \cdot \frac{y - 1}{m-1}$ \\
         $\mathbb{A}_4$: & $\mu(x, y) = \frac{H}{\sqrt{nm}} \cdot \left(\mathbbm{1}_{\{\frac{1}{4} \leq x \leq \frac{1}{2}\}} \mathbbm{1}_{\{0 \leq y \leq \frac{1}{2}\}} + \mathbbm{1}_{\{0 \leq x \leq \frac{1}{4}\}} \mathbbm{1}_{\{\frac{1}{4} \leq y \leq \frac{1}{2}\}}\right)$
    \end{tabular}
    \label{tab-muFuns}
\end{table}

Figure~\ref{fig-alt-all} displays the alternatives presented here, along with markings on how the blocks are constructed. The shifts in $\mathbb{A}_{1}$, $\mathbb{A}_{2}$ and $\mathbb{A}_4$ have a height of $ H \in \{0, 0.05, 0.1, 0.25, 0.5, 1\}$ for independent and $H \in \{0, 0.2, 0.4, 1, 2, 4\}$ for dependent data. In $\mathbb{A}_3$, the ascent is linear with a shift of 0 at $x_{11}, ..., x_{1n}$ up to a shift of $H$ at $x_{n1}, ..., x_{nn}$. Under the hypothesis of no change, we set $\mu \equiv 0$. 
\begin{figure}[ht]
\centering
    \includegraphics[width = 0.6\textwidth]{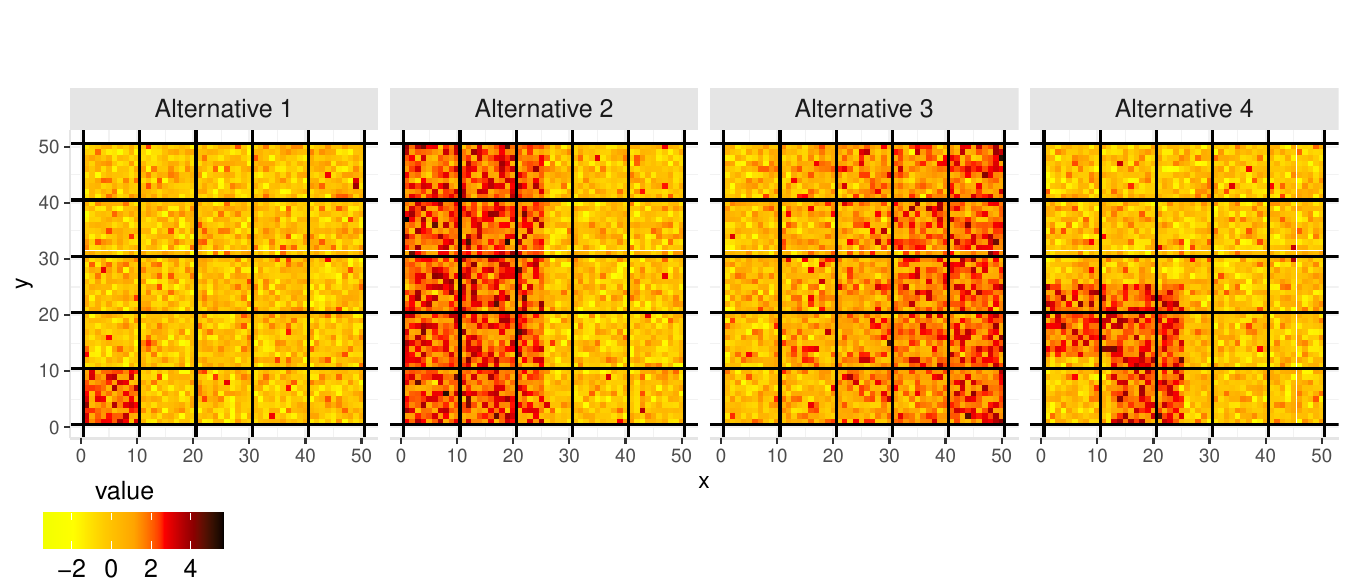}
    \caption{Examples of alternatives}
    \label{fig-alt-all}
\end{figure}
As a variance estimator for scaling the test statistic, we choose the ordinary sample variance 
$\hat{\sigma}^2 = \frac{1}{N-1} \sum_{i = 1}^n \sum_{j = 1}^m (X_{i,j} - \bar{X})^2.$
It is consistent and converges to the true variance $\sigma^2$ in $\mathcal{O}\left(\frac{1}{\sqrt{N}}\right)$ under the hypothesis and appropriate assumptions for the noise, where $N$ is the total number of observations. All results are obtained based on 1000 replications each at the nominal significance level $\alpha = 0.05$.

\subsection{Results under independence}

Table~\ref{tab-size-indep} displays empirical sizes for $n \in \{10, 20, 50\}$ and all three innovation distributions under the null hypothesis, rounded to three digits.
\begin{table}[ht]
\centering
\begin{tabular}{l c rr c rr c rr}
   && \multicolumn{2}{c}{$\mathcal{N}(0, 1)$} && \multicolumn{2}{c}{$t_3$} && \multicolumn{2}{c}{$\chi^2_2$} \\ 
   $n$&& GMD & Var && GMD & Var && GMD & Var\\
\hline
  10   && 0.091 & 0.046 && 0.086 & 0.047 && 0.089 & 0.050 \\  

  20   && 0.058 & 0.056 && 0.048 & 0.049 && 0.048 & 0.053 \\ 

  50   && 0.053 & 0.064 && 0.054 & 0.051 && 0.052 & 0.058 \\ 
\end{tabular}
\caption{Empirical sizes of the GMD and the Var test under the hypothesis at nominal significance level $\alpha = 0.05$ for $n \in \{10, 20, 50\}$ and different innovation distributions, rounded to three digits.} 
\label{tab-size-indep}
\end{table}
The size of  the GMD test exceeds the significance level at $n = 10$ with values between 0.086 and 0.091. For the larger sample sizes considered here, the test keeps the level. For 1000 repetitions, the standard deviation is about $0.0069$, and only for the combinations GMD, $n = 10$ and Var, $\mathcal{N}(0, 1)$, $n = 50$, the empirical sizes are outside of two standard deviations from 0.05. This problem for $n = 10$ does not occur with the Var test and we conclude that the GMD test needs a larger sample size, such as $n, m \geq 20$ to work properly under the hypothesis. All in all, even though the significance level is not seriously infringed, both tests show a slightly liberal behavior.

The three plots in Figure~\ref{fig-resCorr} depict size-corrected power curves of both the GMD and the Var test for $n \in \{10, 20, 50\}$ and all three noise distributions for alternatives $\mathbb{A}_1$ to $\mathbb{A}_4$. 
\begin{figure}[ht]
    \centering
    \includegraphics[page = 1, width=0.6\linewidth]{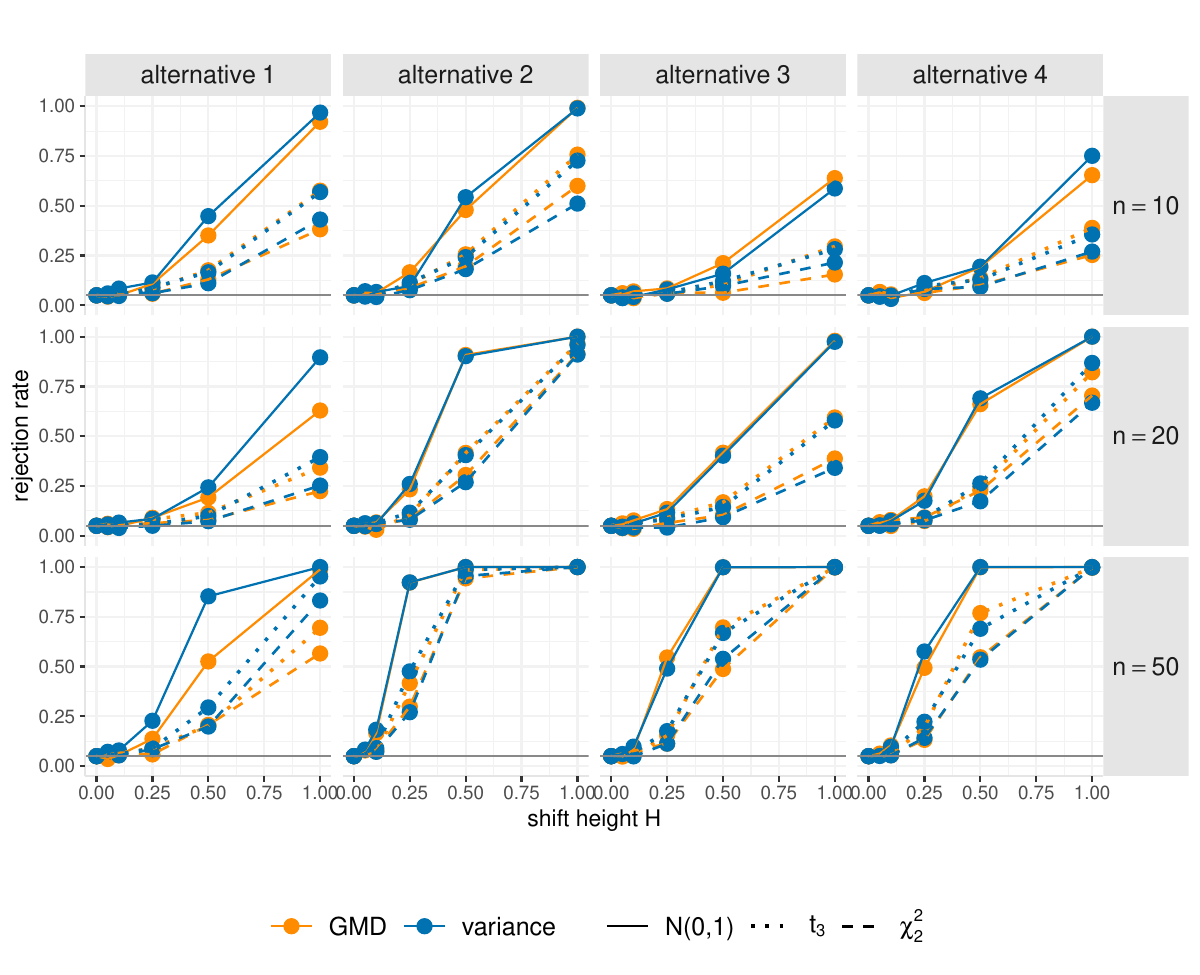}
    \caption{Size-corrected rejection rates of the GMD (orange) and the Var (blue) test at the nominal significance level $\alpha = 0.05$ as a function of location shift $H$, for $n = 10, 20, 50$, alternatives $\mathbb{A}_{1}$ to $\mathbb{A}_{4}$ and $\mathcal{N}(0, 1)$ (solid), $t_3$ (dotted), $\chi^2_2$ (dashed) distributed innovations.}
    \label{fig-resCorr}
\end{figure}
In case of alternative $\mathbb{A}_1$, the Var test shows higher rejection rates for all sample sizes and noise distributions considered here. This is probably due to the higher robustness of the mean difference as a measure of variability, compared to the ordinary variance. As a result, the GMD test ignores the single shifted block in this alternative more often. For $n = 10$ we need to consider that although the size-corrected power curves look similar, the Var test has an advantage over the GMD test as the latter struggles to keep the significance level for this sample size. As expected, both tests have the highest power under $\mathbb{A}_2$ where exactly half of the data is shifted. We conclude that both abrupt changes and trends can be detected quite reliably. Detailed values for shift height $H = 0.5$ can be found in the Appendix in Table~\ref{tab-RRH20}. 
Apart from these findings, we do not detect further relevant differences between the two tests. Even for $t_3$-distributed noise, which does not possess $(4+\varepsilon)$-th moments as required by our asymptotic theory for the Var test, this test still yields good results. Due to the slight advantages of the Var test, we will concentrate on this method in the following. 

\subsection{Results under dependence}

{\color[RGB]{0, 0, 0}
In this study, we choose $n \in \{16, 36, 48\}$, since for the cut-off method we need to make sure that after treating the sample enough observations remain to form a sufficient number of blocks for the test statistic.} Table~\ref{tab-DepH0-2} contains the empirical rejection rates of the Var test for SMA(1) and SAR(1) data with different parameters $\rho$. 
\begin{table}[ht]
\centering
\begin{tabular}{lll c r r r r}
  \hline
  &&n&& \multicolumn{1}{c}{$\rho = 0.1$} & \multicolumn{1}{c}{$\rho = 0.2$} & \multicolumn{1}{c}{$\rho = 0.3$} \\ 
  \hline
  SMA(1) & de-corr.  & 16 && 0.086 & 0.077 & 0.059 \\  
         &           & 36 && 0.033 & 0.030 & 0.034 \\  
         &           & 48 && 0.022 & 0.018 & 0.014 \\  
    \cline{2-7}
        & cut-off    & 16 && 0.030 & 0.030 & 0.026 \\  
        &            & 36 && 0.050 & 0.056 & 0.055 \\  
        &            & 48 && 0.051 & 0.051 & 0.050 \\ 
  \hline
SAR(1) & de-corr.  & 16 && 0.092 & 0.094 & 0.094 \\  
       &           & 36 && 0.037 & 0.056 & 0.156 \\  
       &           & 48 && 0.021 & 0.016 & 0.043 \\  
    \cline{2-7}
   & model residuals & 16 && 0.053 & 0.065 & 0.086 \\  
   &                 & 36 && 0.059 & 0.082 & 0.133 \\  
   &                 & 48 && 0.059 & 0.088 & 0.159 \\ 
  \hline
\end{tabular}
\caption{Empirical sizes of the Var test at nominal significance level $\alpha = 0.05$ for $n \in \{16, 36, 48\}$ and both SMA(1) and SAR(1) random fields with $\rho \in \{0.1, 0.2, 0.3\}$, whitened using de-correlation, the cut-off method and SAR(1) model residuals, under the hypothesis, rounded to three digits.}
\label{tab-DepH0-2}
\end{table}

{\color[RGB]{0, 0, 0}
We notice that the tests have some problems in keeping the significance level in the case of SAR(1) fields with $\rho=0.3$, with empirical sizes up to 0.159.
The autocovariance based de-correlation test keeps the significance level for the largest $n=48$ considered here, while the test using the SAR(1) model residuals also shows some problems for $\rho=0.2$. This might be explained by the bias and large MSE of the sample autocovariances and the parameter estimators for smaller values of $n$ and larger dependency parameters. 
Apart from that, the level is kept very well and can even be seen as slightly conservative for some combinations. The cut-off method approximately maintains the significance level in all scenarios, with values ranging from 0.026 to 0.056. }

Figures~\ref{fig-Dep-MA-2} and \ref{fig-Dep-AR-2} display the size-corrected rejection rates of the Var test for SMA(1) and SAR(1) random fields. 
\begin{figure}
    \centering
    \includegraphics[width=0.6\linewidth, page = 2]{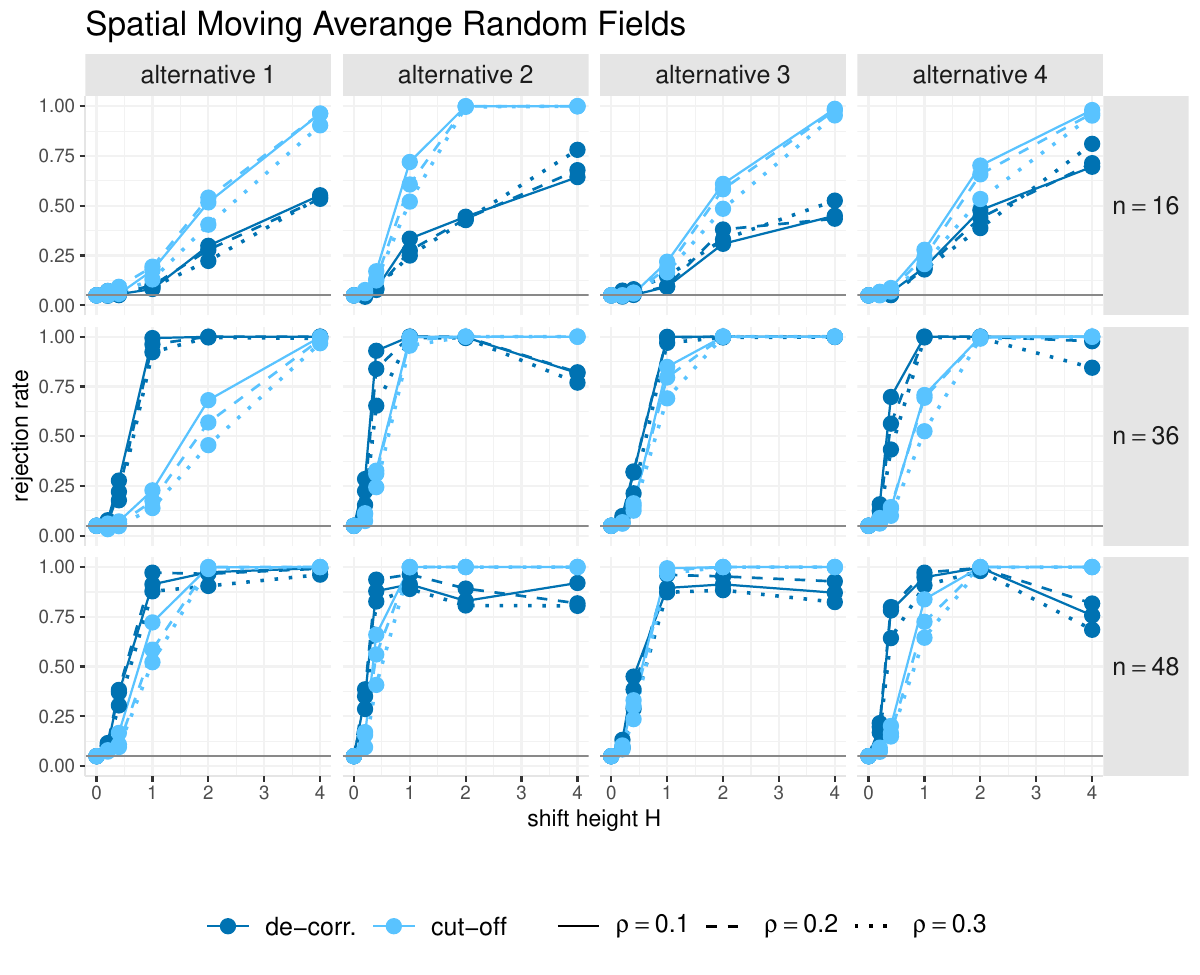}
    \caption{Size-corrected rejection rates of the Var test at the nominal significance level $\alpha = 0.05$ as a function of location shift $H$, for $n = 12, 27, 36$, alternatives $\mathbb{A}_{1}$ to $\mathbb{A}_{4}$, for SMA(1) fields with parameter $\rho \in \{0.1, 0.2, 0.3\}$, whitened by de-correlation (blue) and the cut-off method (light blue).}
    \label{fig-Dep-MA-2}
\end{figure}
\begin{figure}
    \centering
    \includegraphics[width=0.6\linewidth, page = 2]{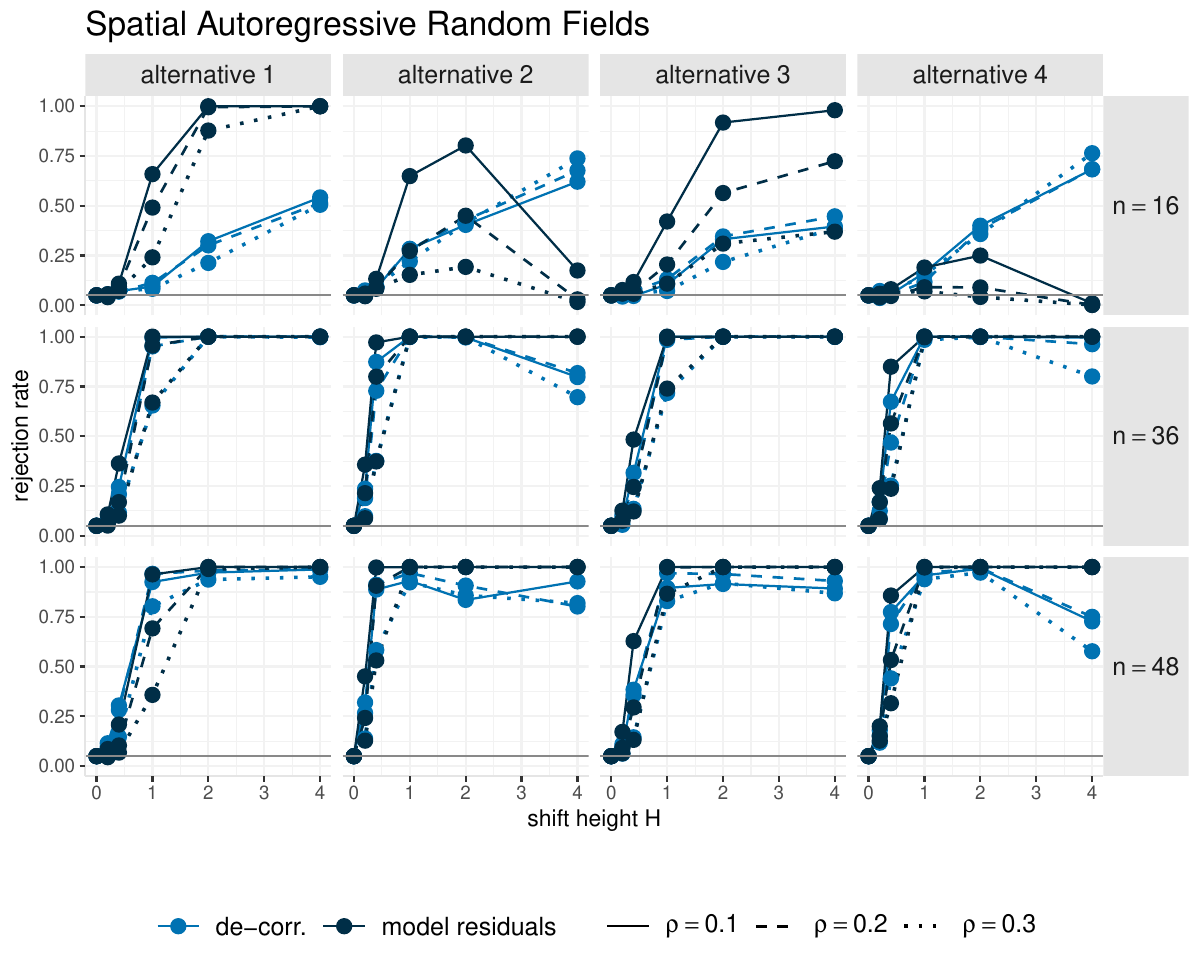}
    \caption{Size-corrected rejection rates of the Var test at the nominal significance level $\alpha = 0.05$ as a function of location shift $H$, for $n = 16, 36, 48$, alternatives $\mathbb{A}_{1}$ to $\mathbb{A}_{4}$, for SAR(1) fields with parameter $\rho \in \{0.1, 0.2, 0.3\}$, whitened by de-correlation (blue) and residuals of a SAR(1) model fit (dark blue).}
    \label{fig-Dep-AR-2}
\end{figure}
{\color[RGB]{0, 0, 0}For the SMA(1) fields, we compare the de-correlation with the cut-off method, whereas for the SAR(1) field, the de-correlation technique and calculation of the SAR(1) model residuals with estimated parameter values are used to whiten the data.} We investigate location shifts up to a height of 4 as opposed to 1 in the previous scenario. 

Considering the SAR(1) scenarios, the first detail striking the eye is that the power curves of the de-correlation and the SAR(1) residuals display a drop for some shift height before regaining power for even higher shifts. {\color[RGB]{0, 0, 0}This lack of monotonicity of the power function has already been discussed in case of change-point tests for time series data; see e.g. \citet{crainiceanu2007nonmonotonic}. It can be explained by the necessity of estimating the autocovariance resp. the dependency parameter. Due to the location changes, this value is overestimated, which leads to a stronger whitening effect and thus to a loss of power}. For larger shifts, this overestimation can be compensated, and we see a convergence to 100\% rejection rates{\color[RGB]{0, 0, 0}, or would see it for even larger sight heights $H$ }. For a simple alternative like $\mathbb{A}_2$ or $\mathbb{A}_4$, the testing procedure  works well already for $n = 16$.  The test on the SAR(1) residuals displays a much more pronounced drop in power under these alternatives $\mathbb{A}_2$ and $\mathbb{A}_4$ for $n = 16$, ultimately making the de-correlation the better method. 
 In all other situations, however, there is little to no loss in power visible and the test on SAR(1) residuals outperforms the test using general de-correlation. 

{\color[RGB]{0, 0, 0} The cut-off method does not have problems with non-monotonic power curves as no parameter estimate is needed. For $n = 16$, the cut-off method performs better than de-correlation, while it is the other way round for the larger sample sizes and a smaller shift height $H < 2$. Nevertheless, as the shift height increases, the de-correlation method is affected by the non-monotonicity and the power drops below that of the cut-off method. In the scenario $\mathbb{A}_1$ and $n = 36$, the lower power of the cut-off method can be explained by the choice of alternative: as there is only a small block affected by the location shift, cutting off data and averaging over the remaining data makes detecting such a small change region difficult. For the remaining scenarios, the cut-off method is a valid competitor to the de-correlation method, especially considering its good behavior under the hypothesis. 

In summary, the methods that use prior information about the dependency structure, i.e., the cut-off method and the test using SAR(1) residuals, mostly yield better power than the general de-correlation method. However, the plots indicate that the latter can still provide good results, without further knowledge about the dependency structure of the data. We also investigated the power of the GMD test applied to the de-correlated data instead of the Var test, but no major differences between the two tests were detected, considering the known difficulties with alternative $\mathbb{A}_1$.}

\section{Application to satellite images} \label{sec-satellite}

The following example illustrates the application of the tests to de-correlated satellite data obtained from the Landsat 8 satellite. This satellite is part of the NASA landsat project and collects data of the earth’s land surface on 9 different spectral bands in the visible and short-wave infrared spectral regions (\citeI{knight2014landsat}). Such satellite imagery helps us to to observe the earth's surface and recognize changes in time. The data can be accessed from \url{https://earthexplorer.usgs.gov/}. To be able to process the data in \texttt{R}, we use the packages \texttt{gdalcubes} (\citeI{gdalcubes1}; \citeI{gdalcubes2}), \texttt{magrittr} (\citeI{magrittr}), \texttt{xts} (\citeI{xts}), \texttt{magick} (\citeI{magick}) and \texttt{tidyverse} (\citeI{tidyverse}). 

We use data from a small region of the Brazilian amazon rainforest, captured on August 12, 2014 and on July 19, 2017. The coordinates of the region are between -7355090 and -7351340 in latitude, and -1023760 and -1019440 in longitude in the EPSG:3857 format. Each pixel comprises a square of $30 \times 30$ meters. In total, we get two images with a size of $144 \times 125$ pixels each.

Our interest is to determine whether deforestation has occurred. To do so, we consider the Normalized Difference Vegetation Index (NDVI). This is a vegetation index measuring the greenness of biomass. It is calculated from the red (visible) and near-infrared spectral bands and takes values between -1 and 1. The greener the biomass, the higher the NDVI. Negative values do not usually occur on land (\citeI{myneni1995interpretation}; \citeI{tucker1979red}).

Figure~\ref{fig-sat1} illustrates the NDVI images captured on two different dates. The earlier left image predominantly displays green areas, with only one very small yellow dot on the right border.
\begin{figure}[ht]
    \centering
    \includegraphics[width = 0.35\textwidth, page = 2]{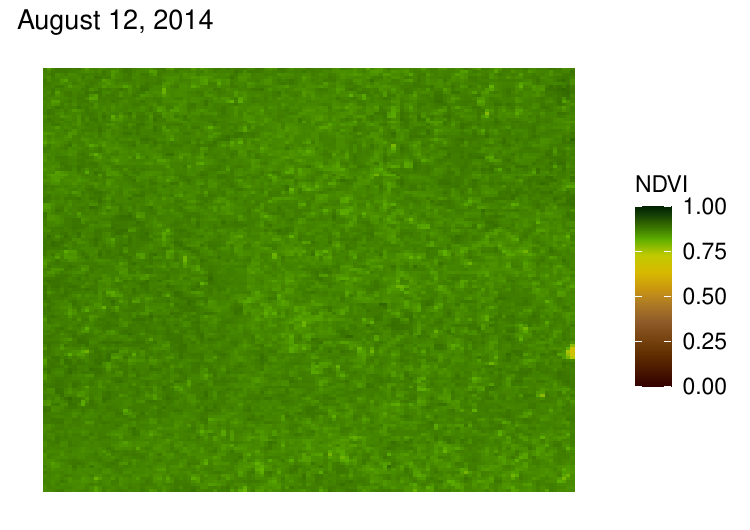}
    \includegraphics[width = 0.35\textwidth, page = 5]{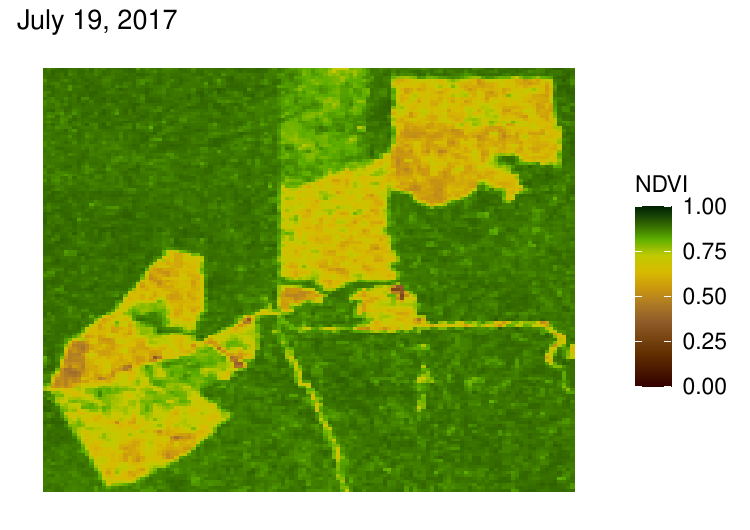}
    \caption{NDVI images of size 144 $\times$ 125 each (left: 2014, right: 2017)}
    \label{fig-sat1}
\end{figure}
In contrast, the right image reveals large yellow-brownish patches that indicate parcels of land where trees have been cut down. {\color[RGB]{0, 0, 0}Since we do not have any information on the dependency structure of the images, we use de-correlation before applying the Var test.} The original images, sized at $144 \times 125$ pixels, are too large to be de-correlated effectively, even under the assumption of a separable covariance function. Doing so would require excessive time and computational resources. Additionally, the "green" in the 2014 image might still show some slight structure, with certain parts appearing darker than others. Therefore, we divide each image into 30 sub-images, where each sub-images measures $24 \times 25$ pixels. Figure~\ref{fig-sat1-split} displays these split images from both 2014 and 2017. 
\begin{figure}[ht]
    \centering
    \includegraphics[width=0.35\textwidth, page = 2]{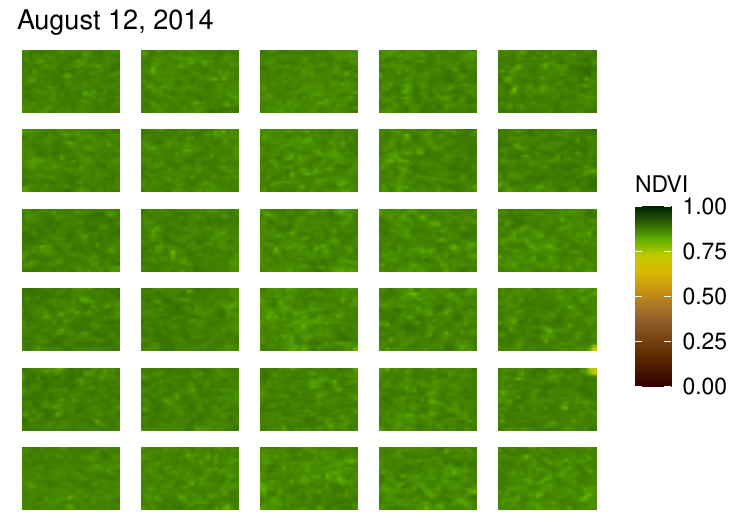}
    \includegraphics[width=0.35\textwidth, page = 5]{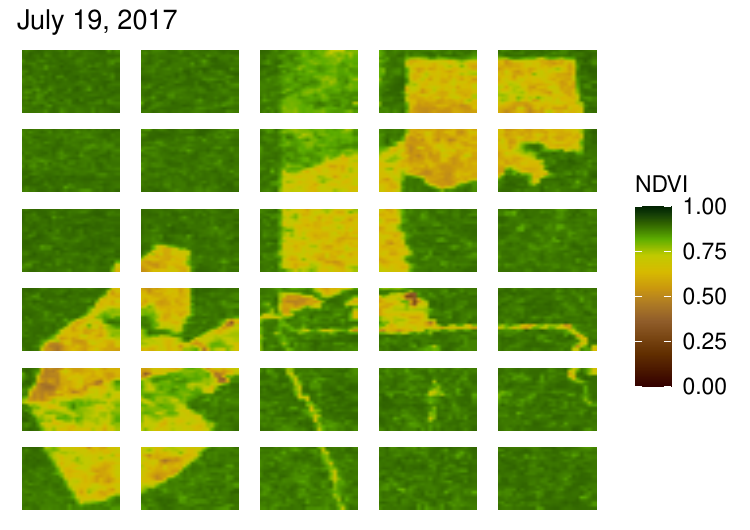}
    \caption{NDVI sub-images after splitting into 5 $\times 6$ images of size $24 \times 25$ each (left: 2014, right: 2017)}
    \label{fig-sat1-split}
\end{figure}
This segmentation simplifies the de-correlation process as we only need to invert 30 matrices of size $600 \times 600$ for each image. The plots in Figure~\ref{fig-sat1-split-decorr} show the de-correlated data from both time points, where each sub-image has been de-correlated individually. In the left plot depicting the data from 2014, we observe only noise without any distinct structure.
\begin{figure}[ht]
    \centering
    \includegraphics[width=0.35\textwidth, page = 2]{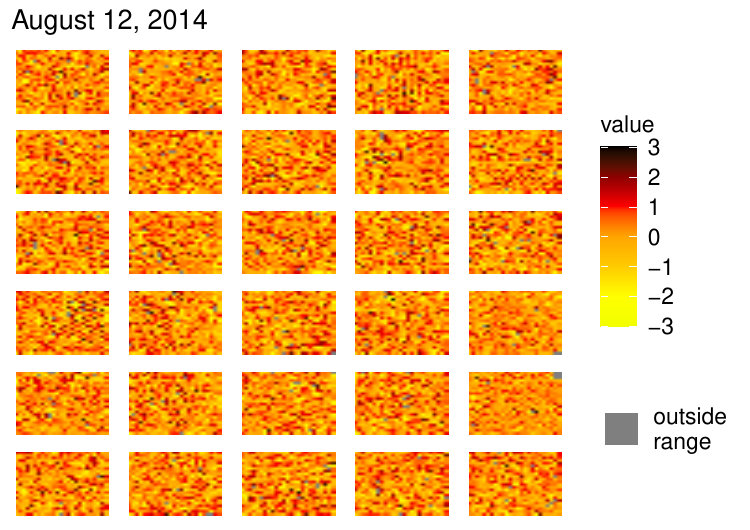}
    \includegraphics[width=0.35\textwidth, page = 5]{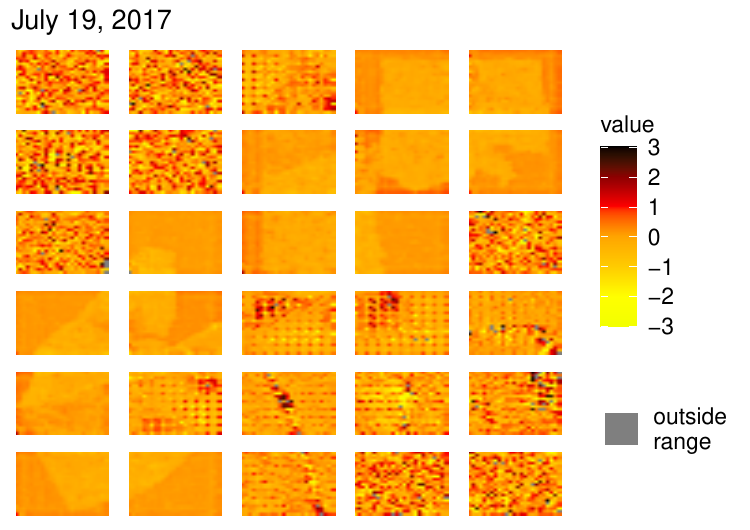}
    \caption{De-correlated sub-images (left: 2014, right: 2017)}
    \label{fig-sat1-split-decorr}
\end{figure}
 Conversely, in the right plot displaying the data from 2017, several deforested areas are clearly visible. Table~\ref{tab-sat-var-decorr} shows the test results of the Var test on the de-correlated sub-images for both dates. Values that are significant at a 5\% level after Bonferroni-Holm correction are displayed in red.
\begin{table}[ht]
\centering
\begin{tabular}{rrrrr}
  \hline
  0.837 & 0.607 & 0.061 & 0.496 & 0.108 \\ 
  0.178 & 0.133 & 0.078 & 0.355 & 0.265 \\ 
  0.299 & 0.648 & 0.084 & 0.639 & 0.826 \\ 
  0.336 & 0.113 & 0.505 & 0.531 & 0.090 \\ 
  0.286 & 0.504 & 0.483 & 0.185 & {\color[RGB]{255, 0, 0}0.000} \\ 
  0.701 & 0.589 & 0.392 & 0.542 & 0.075 \\ 
  \hline
\end{tabular}
\hspace{10mm}
\begin{tabular}{rrrrr}
  \hline
  0.097              & 0.678              & {\color[RGB]{255, 0, 0}0.000} & {\color[RGB]{255, 0, 0}0.000} & {\color[RGB]{255, 0, 0}0.000} \\ 
  0.195              & 0.444              & {\color[RGB]{255, 0, 0}0.000} & {\color[RGB]{255, 0, 0}0.000} & {\color[RGB]{255, 0, 0}0.000} \\ 
  0.724              & {\color[RGB]{255, 0, 0}0.000} & {\color[RGB]{255, 0, 0}0.000} & {\color[RGB]{255, 0, 0}0.000} & 0.534              \\ 
  {\color[RGB]{255, 0, 0}0.000} & {\color[RGB]{255, 0, 0}0.000} & {\color[RGB]{255, 0, 0}0.000} & {\color[RGB]{255, 0, 0}0.000} & 0.261              \\ 
  {\color[RGB]{255, 0, 0}0.000} & {\color[RGB]{255, 0, 0}0.000} & 0.593              & {\color[RGB]{255, 0, 0}0.000} & 0.670              \\ 
  {\color[RGB]{255, 0, 0}0.000} & {\color[RGB]{255, 0, 0}0.000} & 0.496              & 0.765              & 0.362              \\
   \hline
\end{tabular}
\caption{p-values of the Var test on the de-correlated sub-images from August 12, 2014 (left) and Juli 19, 2017 (right), rounded to three digits. Values significant to a 5\% level after Bonferroni-Holm correction are marked in red.} 
\label{tab-sat-var-decorr}
\end{table}

In the left table, there is only one significant p-value, corresponding to the sub-image at position (5, 5) containing part of the visible dot, while the test results are not significant for the other sub-images. As opposed to this, the right table displays 18 significant values. For all sub-images with a significant p-value, we can clearly see deforestation in the original images (Figure~\ref{fig-sat1-split}). Only four images, namely at positions (6, 3), (5, 3), (4, 5) and (5, 5), show signs of deforestation caused by roads, but the test does not reject the hypothesis. This is expected as by construction our test detects change regions with positive Lebesgue measure consistently, but not lines. For all sub-images without deforestation, the test correctly did not reject the hypothesis.

\section{Summary} \label{sec-summary}

We have introduced two tests for arbitrary changes in location for random fields, the GMD test extension of the method of \mycite{schmidt2024detecting} for two-dimensional data, and the Var test as an extension of the classical ANOVA. The tests assume independent observations and the existence of $2 + \varepsilon$ (GMD test) resp. $4 + \varepsilon$ (Var test) central moments. We have shown the convergence of the test statistic of the Var test to a normal distribution, and its consistency against  change regions with positive Lebesgue-measure. For the GMD test, both these properties can be deduced from the proofs in \mycite{schmidt2024detecting}. In a simulation study, we have demonstrated that both tests can be successfully applied to correlated data after {\color[RGB]{0, 0, 0}whitening}. In an application to satellite images, we showed that the Var test can reliably detect regions affected by deforestation.

\section*{Acknowledgments}
This research was (partially) funded in the course of TRR 391 Spatio-temporal Statistics for the
Transition of Energy and Transport (520388526) by the Deutsche Forschungsgemeinschaft (DFG,
German Research Foundation).\\
The authors gratefully acknowledge the computing time provided on the Linux HPC cluster at TU Dortmund University (LiDO3), partially funded in the course of the Large-Scale Equipment Initiative by the Deutsche Forschungsgemeinschaft (DFG, German Research Foundation) as project 271512359.

\bibliographystyle{elsarticle-harv} 
\bibliography{cas-refs}

\appendix

\renewcommand\thefigure{\Alph{section}.\arabic{figure}}    
\renewcommand\thetable{\Alph{section}.\arabic{table}} 

\newpage
\section{Detailed proofs}
\setcounter{figure}{0}   
\setcounter{table}{0}   

\begin{proof}[Proof of Proposition \ref{prop-CLT}]
    Let $h, k \in \{1, ..., b_n\} \times \{1, ..., b_m\}$ be fixed. We notice that
    $$(Y^{(n,m)}_{i,j}) = \left(Y^{(n,m)}_{i,j}; \, i, j \in \mathcal{I}_{h,k}; \, n,m \in \mathbb{N}\right)$$ is a double array. For every combination of $n, m \in \mathbb{N}$ the random variables $\left(Y^{(n,m)}_{i,j}\right)_{i,j}$ are independent. When defining $S_{h,k}^{(n,m)} := \doublesumBlock Y_{i,j}$, we can rewrite the statistic as
    $$ \sqrt{l_n l_m} \,\frac{\hat{\nu}_{h,k}^{(n,m)}}{\sigma} = \frac{S_{h,k}^{(n,m)}}{\sqrt{l_n l_m \sigma^2}}$$
    and it holds that
    \begin{align*}
        \ew\left(S_{h,k}^{(n,m)}\right) = \ew\left(\doublesumBlock Y_{i,j} \right)  {\color[RGB]{0, 0, 0}= 0, \quad}
        \Var\left(S_{h,k}^{(n,m)}\right) = \Var\left(\doublesumBlock Y_{i,j} \right) = l_n l_m \sigma^2.
    \end{align*}
    For the proposition to hold, we need to verify Lyapunov's condition. Denote with $M_{k}$ the $k$-th central moment of $Y_{i,j}$, then we get
    \begin{align*}
        \frac{1}{\Var\left(S_{h,k}^{(n,m)} \right)^{1 + \frac{\delta}{2}}} \doublesumBlock \ew \left(\left|  Y_{i,j} \right|^{2 + \delta} \right) 
        = \frac{1}{\left(l_n l_m \sigma^2\right)^{1 + \frac{\delta}{2}}} \doublesumBlock M_{2+\delta} 
        = \mathcal{O}\left(\frac{1}{(l_n l_m)^{\frac{\delta}{2}}}\right).
    \end{align*}
\end{proof}

\begin{proof}[Proof of Lemma~\ref{lemm-4thM}]
    First, we notice that according to Proposition~\ref{prop-CLT}, $\frac{\sqrt{nm} \bar{Y}_{n,m}}{\sigma} \overset{\mathcal{D}}{\longrightarrow} \mathcal{N}(0, 1)$ and with the Continuous Mapping Theorem, $\left[\frac{\sqrt{nm} \bar{Y}_{n,m}}{\sigma}\right]^4 \overset{\mathcal{D}}{\longrightarrow} \mathcal{N}(0, 1)^4$.\vspace{1mm}\\
    Next we show that $\left\{\frac{(nm)^2 \bar{Y}_{n,m}^4}{\sigma^4} \, \big| \, n \geq 1, \, m \geq 1\right\}$ is uniformly integrable. {\color[RGB]{0, 0, 0}From the M-Z inequality we get}
    $$\ew\left[\left(\frac{\left(\sum_{i = 1}^n \sum_{j = 1}^m Y_{i,j}\right)^4}{(nm)^2 \sigma^4}\right)^{1 + \delta} \right] < C < \infty, $$
    for a positive constant $C$. According to the postscript to Theorem~5.3 in \mycitep{billingsley1968convergence}{32}, $\left\{\frac{(nm)^2 \bar{Y}_{n,m}^4}{\sigma^4} \, \big| \, n \geq 1, \, m \geq 1\right\}$ is therefore uniformly integrable. We can apply Theorem~5.4 from \mycitep{billingsley1968convergence}{32} and conclude that 
    $$\kappa_{n,m}^{(4)} = \ew\left[\frac{(nm)^2 \bar{Y}_{n,m}^4}{\sigma^4}\right] \longrightarrow \ew(Z^4), \quad Z \sim \mathcal{N}(0, 1).$$
    The fourth moment of a standard normal distribution is known to be 3.
\end{proof}

\begin{proof}[Verification of Lyapunov's condition in Theorem~\ref{th-LyapunovCLT-withoutCentering}]
We need to verify that
\begin{align*}
    \lim_{n \rightarrow \infty} \frac{1}{\left( \sum_{h = 1}^{b_n} \sum_{k = 1}^{b_m} \Var(\nutilde^2) \right)^{1 + \frac{\varepsilon}{4}}} \sum_{h = 1}^{b_n} \sum_{k = 1}^{b_m} \ew\left[\left| \nutilde^2 - \sigma^2 \right|^{2 + \frac{\varepsilon}{2}} \right] = 0.
\end{align*}
Since $\mutilde^2$ and $\sigma^2$ are both positive (first step), by the $c_r$-inequality (second step), and by using again {\color[RGB]{0, 0, 0} the M-Z inequality} (fourth step), we get for the numerator 
\begin{align*}
    \sum_{h = 1}^{b_n} \sum_{k = 1}^{b_m}\ew\left[\left| \nutilde^2 - \sigma^2 \right|^{2 + \frac{\varepsilon}{2}} \right]
    &\leq \sum_{h = 1}^{b_n} \sum_{k = 1}^{b_m} \ew\left[\left| \nutilde^2 + \sigma^2 \right|^{2 + \frac{\varepsilon}{2}} \right] \\
    &\leq const \cdot \sum_{h = 1}^{b_n} \sum_{k = 1}^{b_m} \left(\ew\left( |\nutilde|^{4 + \varepsilon} \right) + \sigma^{4 + \varepsilon}\right)\\
    &= const \cdot \sum_{h = 1}^{b_n} \sum_{k = 1}^{b_m} \left(\frac{1}{(l_n l_m)^{2 + \frac{\varepsilon}{2}}} \ew\left(\left|  \doublesumBlock Y_{i,j} \right|^{4 + \varepsilon} \right) + \sigma^{4 + \varepsilon}\right)\\
    &< const \cdot b_n b_m
\end{align*}
for $n,m$ large enough and $\ew(|Y_{1,1}|^{4 + \varepsilon}) < \infty$, as assumed. Now considering the denominator, we have
$$\sum_{h = 1}^{b_n} \sum_{k = 1}^{b_m} \Var\left(\nutilde^2\right) = b_n b_m \cdot \sigma^4 (\kappa_{l_n, l_m}^{(4)} - 1),$$
so in total, we get
\begin{align*}
    \frac{1}{\left( \sum_{h = 1}^{b_n} \sum_{k = 1}^{b_m} \Var(\nutilde^2) \right)^{1 + \frac{\varepsilon}{4}}} &\sum_{h = 1}^{b_n} \sum_{k = 1}^{b_m} \ew\left[\left| \nutilde^2 - \sigma^2 \right|^{2 + \frac{\varepsilon}{2}} \right]
    < \frac{const}{(b_n b_m)^{\frac{\varepsilon}{4}}} \Big(\underbrace{\kappa_{l_n, l_m}^{(4)} - 1}_{\rightarrow 2}\Big)^{-(1 + \frac{\varepsilon}{4})}
    = \mathcal{O}\left(\frac{1}{(b_n b_m)^{\frac{\varepsilon}{4}}}\right).
\end{align*}

\end{proof}    

\newpage
\section{Tables and figures}

\begin{table}[ht]
\centering
\begin{tabular}{l c rr c rr c rr}
  \hline
   && \multicolumn{2}{c}{$\mathcal{N}(0, 1)$} && \multicolumn{2}{c}{$t_3$} && \multicolumn{2}{c}{$\chi^2_2$} \\ 
   && GMD & Var && GMD & Var && GMD & Var\\
  \hline
  &&\multicolumn{8}{c}{$n = 10$} \\
  \hline
  $\mathbb{A}_{1}$ && 0.482 & 0.409 && 0.253 & 0.172 && 0.216 & 0.117 \\ 
  $\mathbb{A}_{2}$ && 0.618 & 0.538 && 0.335 & 0.229 && 0.249 & 0.163 \\ 
  $\mathbb{A}_{3}$ && 0.251 & 0.175 && 0.136 & 0.098 && 0.135 & 0.089 \\ 
  $\mathbb{A}_{4}$ && 0.378 & 0.288 && 0.227 & 0.138 && 0.150 & 0.093 \\ 
   \hline
  &&\multicolumn{8}{c}{$n = 20$} \\
  \hline
  $\mathbb{A}_{1}$ && 0.250 & 0.257 && 0.103 & 0.114 && 0.098 & 0.105 \\ 
  $\mathbb{A}_{2}$ && 0.938 & 0.925 && 0.400 & 0.439 && 0.300 & 0.307 \\ 
  $\mathbb{A}_{3}$ && 0.413 & 0.389 && 0.124 & 0.148 && 0.123 & 0.124 \\ 
  $\mathbb{A}_{4}$ && 0.688 & 0.687 && 0.253 & 0.261 && 0.169 & 0.210 \\  
   \hline
  &&\multicolumn{8}{c}{$n = 50$} \\
  \hline
  $\mathbb{A}_{1}$ && 0.580 & 0.846 && 0.228 & 0.318 && 0.166 & 0.238 \\ 
  $\mathbb{A}_{2}$ && 1.000 & 1.000 && 0.988 & 0.987 && 0.949 & 0.948 \\ 
  $\mathbb{A}_{3}$ && 0.999 & 1.000 && 0.687 & 0.718 && 0.548 & 0.535 \\ 
  $\mathbb{A}_{4}$ && 1.000 & 1.000 && 0.712 & 0.746 && 0.565 & 0.578 \\  
   \hline
\end{tabular}
\caption{Simulated rejection rates of the GMD and the Var test at the nominal significance level $\alpha = 0.05$ for a location shift $H = 0.5$, $n \in \{10, 20, 50\}$, and different innovation distributions under the alternatives $\mathbb{A}_{1}$ to $\mathbb{A}_{4}$, rounded to three digits.} 
\label{tab-RRH20}
\end{table}

\begin{table}[ht]
\centering
\begin{tabular}{lrrrrrrrrr}
 &&& \multicolumn{3}{c}{SMA(1)} && \multicolumn{3}{c}{SAR(1)}\\
  \hline
 &&& $\rho = 0.1$ & $\rho = 0.2$ & $\rho = 0.3$ && $\rho = 0.1$ & $\rho = 0.2$ & $\rho = 0.3$  \\ 
  \hline
  &&&\multicolumn{7}{c}{$n = 16$} \\
  \cline{2-10}
  de-correlation & $\mathbb{A}_{1}$ && 0.410 & 0.364 & 0.270 && 0.419 & 0.373 & 0.286 \\  
            & $\mathbb{A}_{2}$ && 0.506 & 0.494 & 0.435 && 0.533 & 0.492 & 0.413 \\  
            & $\mathbb{A}_{3}$ && 0.423 & 0.449 & 0.338 && 0.428 & 0.425 & 0.317 \\  
            & $\mathbb{A}_{4}$ && 0.525 & 0.505 & 0.418 && 0.500 & 0.491 & 0.417 \\  
  \cline{2-10}
  &&&\multicolumn{7}{c}{$n = 36$} \\
  \cline{2-10}
            & $\mathbb{A}_{1}$ && 1.000 & 1.000 & 0.996 && 1.000 & 1.000 & 1.000 \\  
            & $\mathbb{A}_{2}$ && 0.999 & 1.000 & 0.991 && 0.998 & 0.999 & 0.996 \\  
            & $\mathbb{A}_{3}$ && 1.000 & 1.000 & 0.994 && 0.999 & 1.000 & 1.000 \\  
            & $\mathbb{A}_{4}$ && 1.000 & 1.000 & 0.994 && 1.000 & 1.000 & 1.000 \\  
  \cline{2-10}
  &&&\multicolumn{7}{c}{$n = 48$} \\
  \cline{2-10}
            & $\mathbb{A}_{1}$ && 0.957 & 0.974 & 0.914 && 0.965 & 0.973 & 0.944 \\  
            & $\mathbb{A}_{2}$ && 0.815 & 0.898 & 0.808 && 0.822 & 0.908 & 0.853 \\  
            & $\mathbb{A}_{3}$ && 0.905 & 0.950 & 0.868 && 0.906 & 0.958 & 0.916 \\  
            & $\mathbb{A}_{4}$ && 0.997 & 0.996 & 0.967 && 0.993 & 0.992 & 0.964 \\  
   \hline
  &&&\multicolumn{7}{c}{$n = 16$} \\
  \cline{2-10}
  cut-off (SMA) & $\mathbb{A}_{1}$ & & 0.453 & 0.379 & 0.317  && 1.000 & 0.995 & 0.878 \\ 
  model residuals (SAR) &   $\mathbb{A}_{2}$ & & 1.000 & 0.996 & 0.984     &&       0.803 & 0.450 & 0.193 \\  
   &   $\mathbb{A}_{3}$ & & 0.550 & 0.469 & 0.407     &&         0.918 & 0.564 & 0.310 \\ 
   &   $\mathbb{A}_{4}$ & & 0.570 & 0.513 & 0.447     &&         0.250 & 0.090 & 0.042 \\ 
   \cline{2-10}
  &&&\multicolumn{7}{c}{$n = 36$} \\
  \cline{2-10}
   &   $\mathbb{A}_{1}$ & & 0.648 & 0.570 & 0.508     &&         1.000 & 1.000 & 1.000 \\   
   &   $\mathbb{A}_{2}$ & & 1.000 & 1.000 & 1.000     &&         1.000 & 1.000 & 1.000 \\   
   &   $\mathbb{A}_{3}$ & & 1.000 & 1.000 & 1.000     &&         1.000 & 1.000 & 1.000 \\   
   &   $\mathbb{A}_{4}$ & & 1.000 & 0.997 & 0.987     &&         1.000 & 1.000 & 1.000 \\  
   \cline{2-10}
  &&&\multicolumn{7}{c}{$n = 48$} \\
  \cline{2-10}
   &   $\mathbb{A}_{1}$ & & 0.999 & 0.995 & 0.993     &&         1.000 & 1.000 & 0.990 \\   
   &   $\mathbb{A}_{2}$ & & 1.000 & 1.000 & 1.000     &&         1.000 & 1.000 & 1.000 \\   
   &   $\mathbb{A}_{3}$ & & 1.000 & 1.000 & 1.000     &&         1.000 & 1.000 & 1.000 \\   
   &   $\mathbb{A}_{4}$ & & 1.000 & 1.000 & 0.996    &&         1.000 & 1.000 & 1.000 \\  
  \hline
\end{tabular}
\caption{Simulated rejection rates of the Var test at the nominal significance level $\alpha = 0.05$ for a location shift $H = 2$, $n \in \{16, 36, 48\}$, SMA(1) and SAR(1) models with parameters $\rho \in \{0.1, 0.2, 0.3\}$ whitened using de-correlation, the cut-off method and SAR(1) model residuals, under the alternatives $\mathbb{A}_{1}$ to $\mathbb{A}_{4}$, rounded to three digits.}
\label{tab-RRH2dep}
\end{table}

\end{document}